\documentclass[11pt,a4paper]{article}

\usepackage[T1]{fontenc}
\usepackage[utf8]{inputenc}
\usepackage[margin=1in]{geometry}
\usepackage{graphicx}
\usepackage{amsmath,amssymb}
\usepackage{amsthm}
\usepackage{booktabs}
\usepackage{multirow}
\usepackage{longtable}
\usepackage{tabularx}
\newcolumntype{L}[1]{>{\raggedright\arraybackslash}p{#1}}
\usepackage{url}
\usepackage{listings}
\usepackage{xcolor}
\usepackage{float}
\usepackage[hidelinks,bookmarksnumbered]{hyperref}
\makeatletter
\g@addto@macro\UrlBreaks{\do\_\do\/}
\makeatother
\usepackage{tikz}
\usetikzlibrary{arrows.meta,positioning}

\theoremstyle{plain}
\newtheorem{theorem}{Theorem}[section]
\newtheorem{lemma}[theorem]{Lemma}
\newtheorem{corollary}[theorem]{Corollary}
\theoremstyle{definition}
\newtheorem{definition}[theorem]{Definition}
\newcommand{\resHolds}{\textsc{holds}}
\newcommand{\resViolated}{\textsc{violated}}
\newcommand{\resUnknown}{\textsc{unknown}}
\newcommand{\solverSAT}{\textsc{sat}}
\newcommand{\solverUNSAT}{\textsc{unsat}}

\lstdefinelanguage{dsltrans}{
  keywords={metamodel, transformation, layer, rule, match, apply, backward,
            property, precondition, postcondition, any, direct, class, abstract,
            extends, enum, where, Bool, Int, String, true, false},
  sensitive=true,
  morecomment=[l]{//},
  morestring=[b]",
  basicstyle=\ttfamily\small,
  keywordstyle=\bfseries\color{blue!70!black},
  commentstyle=\itshape\color{gray},
  stringstyle=\color{red!70!black},
  showstringspaces=false,
  tabsize=2,
  breaklines=true,
  frame=single,
  numbers=left,
  numberstyle=\tiny\color{gray},
  xleftmargin=2em,
  framexleftmargin=1.5em
}

\title{Tractable Verification of Model Transformations:\\
A Cutoff-Theorem Approach for DSLTrans}

\author{%
  Levi L\'ucio\\
  \textit{Airbus Defense and Space}, Germany\\
  \texttt{levi.lucio@airbus.com}%
}

\date{}

\begin{document}
\maketitle

\begin{abstract}
Model transformations are central to Model-Driven Engineering (MDE), yet their formal verification has long been hindered by the undecidability of general-purpose transformation languages. DSLTrans was designed as a Turing-incomplete transformation language, intentionally trading expressiveness for verifiability. Despite this design philosophy, previous verification attempts based on path condition enumeration suffered from exponential blow-up and failed to scale to realistic transformations.

This paper presents a \emph{practically tractable verification workflow} for DSLTrans transformations and explains the formal conditions under which that workflow is complete. The workflow rests on three mutually reinforcing contributions: (i) a \emph{Cutoff Theorem} that makes bounded model checking complete for a precise fragment of DSLTrans and positive existence/traceability properties, reducing verification from an infinite search space to a finite, computable bound; (ii) a suite of composable, soundness-preserving optimizations---per-class bounds, fragment-based verification with CEGAR, and trace-aware dependency analysis---that substantially reduce the SMT encoding size; and (iii) an implementation backed by Z3 and validated on a concrete corpus of realistic model transformations drawn from the ATL Zoo and related community benchmarks.

We evaluate the approach on a corpus of 29 concrete transformations with 899 properties, spanning compiler lowering, schema translation, behavioral modeling, graph mapping, and synthetic stress tests. Of these, 552 properties are proved, 345 produce concrete counterexamples---a number that includes intentional negative tests and boundary properties that mark known non-guarantees of the current transformations, not only unintended bugs---and only 2 remain undecided within the timeout budget. For properties that exceed the tractability budget, we introduce a \emph{tractability-driven property refinement} technique---precondition specialization, postcondition decomposition, and transformation instrumentation---that achieves up to 112$\times$ speedup while eliminating spurious counterexamples. The approach is supported by a web-based IDE for authoring transformations and properties, and by a concrete execution engine for runtime validation.
\end{abstract}


\medskip

\section{Introduction}
\label{sec:introduction}

Model-Driven Engineering (MDE) relies on model transformations as first-class artifacts that translate models conforming to one metamodel into models conforming to another~\cite{sendall2003model,czarnecki2006feature}. As transformations become more complex and are deployed in safety-critical domains such as aerospace, automotive, and medical systems, the need for formal guarantees about their correctness grows correspondingly~\cite{cabot2010verification,lucio2010validation}.

Unfortunately, the formal verification of model transformations is inherently difficult. Most widely used transformation languages---such as ATL~\cite{jouault2008atl}, QVT~\cite{omg2016qvt}, and Epsilon~\cite{kolovos2008epsilon}---are Turing-complete, meaning that verification of arbitrary properties is undecidable in general. Even for restricted property classes, the combinatorial explosion of possible input models renders exhaustive verification intractable for practical transformations.

\textbf{DSLTrans} was designed from the outset to address this tension~\cite{barroca2011dsltrans}. It is a \emph{Turing-incomplete} out-place model transformation language that deliberately sacrifices expressiveness---forbidding unbounded recursion, in-place updates, and negative application conditions---in exchange for amenability to formal analysis. The language enforces a layered execution model with monotonic semantics: rules only add structure to the target model, never remove or modify existing elements.

\subsection{Previous Verification Attempts}

Early verification approaches for DSLTrans were based on \emph{symbolic path condition enumeration}~\cite{lucio2010validation,oakes2015fully,oakes2018thesis}. The idea was to enumerate all possible execution paths through the transformation---each path condition representing a distinct combination of rule firings---and then check each path against the property being verified. While theoretically sound, this approach suffered from a fundamental scalability problem: the number of path conditions grows exponentially with the number of rules and match patterns, reaching millions or billions for transformations with more than a handful of rules.

Several optimization techniques were proposed to mitigate this explosion, including path condition subsumption and merging~\cite{oakes2018full}. Despite these efforts, the approach could not be applied to transformations of realistic size. Benchmark transformations with 15--20 rules and 10+ properties routinely exceeded memory limits or required hours of computation, rendering the approach impractical for the iterative development workflows that engineers require.

\paragraph{What is new in this paper relative to that prior line of work.}
The first author of the present paper was also a co-author of the original DSLTrans verification approach~\cite{lucio2010validation} and of its symbolic-execution successor~\cite{oakes2015fully,oakes2018full}, and we want to be explicit about what is and is not inherited from that work. The earlier line of work fixed a Turing-incomplete language design and a sound path-condition enumeration semantics, and developed optimizations (subsumption, merging) that addressed scalability \emph{within} the path-condition paradigm. Crucially, no cutoff theorem was proposed, conjectured, or actively pursued in that prior work---the question of whether a small-model property could be derived from the structure of an F-LNR transformation was simply not on the agenda. In hindsight, the path-condition approach could only have scaled to realistic transformations through extreme compression of path conditions, an avenue that the literature did not find. The contributions in the present paper---the F-LNR/G-BPP fragment characterization, the cutoff theorem with its three concrete bound formulas, the per-class fixed-point bound, the trace-aware soundness lemma, the SMT-direct encoding, and the composability theorem---are therefore genuinely new with respect to the earlier DSLTrans line, and represent a different verification architecture rather than an incremental optimization of the previous one.

\subsection{Contributions of This Paper}

This paper presents a fundamentally different verification strategy that makes formal property checking of DSLTrans transformations practically tractable. Rather than enumerating paths, we encode the transformation semantics directly into an SMT formula and leverage the \emph{Z3} solver~\cite{demoura2008z3} for bounded model checking. The key insight is the \textbf{Cutoff Theorem}: for the fragment of DSLTrans we consider and for positive existence/traceability properties, every property violation has a bounded witness whose size depends only on the static characteristics of the transformation and property, not on the input model. This transforms bounded model checking from a heuristic into a \emph{complete verification method} for the stated fragment.

The primary contribution is a practically tractable verification workflow for DSLTrans; the cutoff theorem provides the formal basis for when bounded checking is complete, and the optimization suite makes completeness practically reachable.

Our contributions are:

\begin{enumerate}
  \item \textbf{A tractable verification workflow:} We implement a Python-based prover using Z3 that verifies properties of concrete DSLTrans transformations via automatic attribute abstraction and SMT-based bounded model checking. We evaluate it on a corpus of 29 concrete transformations with 899 properties and show that verification workloads remain tractable in practice.

  \item \textbf{The Cutoff Theorem for F-LNR $\times$ G-BPP:} We prove that for the Local, Non-Recursive fragment of DSLTrans (F-LNR) and the class of Bounded Positive Pattern properties (G-BPP), every violation has a small counterexample. We derive three alternative bounds ($K_{\text{coarse}}$, $K_{\text{sharp}}$, $K_{\text{tight}}$) and prove their soundness. This theorem is the formal justification for when bounded checking is complete.

  \item \textbf{Composable, soundness-preserving optimizations:} We develop per-class bounds via fixed-point analysis, minimal fragment selection with CEGAR refinement, trace-and-attribute-aware dependency pruning, and equisatisfiable SMT encodings. These compose safely and reduce verification time by factors of 10--1000$\times$ on our benchmarks.

  \item \textbf{A web-based IDE:} We provide a browser-based development environment for authoring transformations, defining properties, executing transformations concretely, and running formal verification---supporting an interactive, test-driven development workflow. The deployed instance is publicly available online~\cite{dsltrans2026webapp}.

  \item \textbf{Methodological insights:} We report on the iterative co-evolution of mathematical theory and implementation, and on the test-driven methodology for defining properties and constructing transformations.
\end{enumerate}

\subsection{Paper Organization}

The remainder of the paper is organized as follows:
\begin{itemize}
  \item Section~\ref{sec:worked-example} introduces the running UML-to-Java example.
  \item Section~\ref{sec:dsltrans} defines the DSLTrans fragment and the property fragment for which the results apply.
  \item Section~\ref{sec:cutoff} presents the cutoff theorem and its claim boundary.
  \item Section~\ref{sec:optimizations} develops the soundness-preserving optimization stack.
  \item Section~\ref{sec:prover} describes the prover architecture, implementation, and empirical evaluation.
  \item Section~\ref{sec:tractability} explains why the resulting verification problem is tractable in practice.
  \item Section~\ref{sec:methodology}, Section~\ref{sec:related-work}, and Section~\ref{sec:threats} discuss methodology, related work, and threats to validity.
  \item Section~\ref{sec:conclusion} concludes.
\end{itemize}

\section{Running Example: UML to Java}
\label{sec:worked-example}

To ground the discussion, we use as running example a transformation from UML class diagrams to Java Abstract Syntax Trees (ASTs). This transformation is a standard benchmark in the model transformation community, frequently used to evaluate transformation languages and verification approaches~\cite{atl-zoo}. Our version is substantive: it comprises 6 layers, 22 rules, and 29 properties, covering package mapping, classifier-to-compilation-unit creation, member generation (fields, methods, constructors), accessor synthesis (getters/setters), inter-type relationships (extends, implements, imports), and boilerplate generation (toString, equals, hashCode).

We chose this example because it exercises every aspect of the verification method: (i) \emph{layered structure}: six layers with backward links across layers test the dependency-depth analysis; (ii) \emph{trace-link richness}: most properties require trace links connecting source and target elements, exercising trace-aware relevance pruning; (iii) \emph{type polymorphism}: the metamodel contains abstract classes with multiple concrete subtypes, stressing per-class bound analysis and inheritance flattening; (iv) \emph{structural complexity}: accessor and boilerplate rules produce multiple target elements per match, leading to nontrivial cutoff bounds; and (v) \emph{community familiarity}: UML-to-Java variants appear across the ATL Zoo, TTC benchmarks, and related literature, providing external comparability.

\subsection{DSLTrans at a Glance}
\label{subsec:dsltrans-glance}

Before presenting the example, we briefly outline the key constructs of DSLTrans to help the reader parse the listings and diagrams that follow. A complete formalization appears in Section~\ref{sec:dsltrans}.

A DSLTrans specification is a single \texttt{.dslt} file containing \emph{metamodel}, \emph{transformation}, and \emph{property} declarations. A \textbf{metamodel} declares the typed classes, attributes, and associations of a modeling language. A \textbf{transformation} maps a source metamodel to a target metamodel via an ordered sequence of \emph{layers}, each containing one or more \emph{rules}. Each \textbf{rule} has three blocks: (i)~a \texttt{match} block, which specifies a pattern of source elements to find---elements prefixed \texttt{any} are universally quantified, and \texttt{direct} associations constrain structural adjacency; (ii)~an \texttt{apply} block, which specifies target elements and links to create when the match succeeds, optionally binding target attributes from matched source attributes; and (iii)~an optional \texttt{backward} block, whose concrete syntax uses \texttt{<{-}-trace-{-}>} to denote \emph{backward links} (i.e., dependencies on traces produced by earlier layers). A \textbf{property} is a \emph{precondition--postcondition} pair: for every match of the precondition pattern in the source model, a corresponding postcondition pattern must exist in the target model. Postconditions may include trace links to assert that specific target elements were produced from specific source elements.

DSLTrans also has a graphical syntax, rendered by the web-based IDE (Figure~\ref{fig:webapp-uml-java-transformation}). In the graphical view, each rule is shown as a node-and-link diagram with match elements above and apply elements below, while backward links are drawn as vertical connectors linking apply elements to previously created target elements. The textual listings in this paper are the authoritative specification; the graphical views provide a complementary visual overview of rule structure.

\subsection{Source and Target Metamodels}

The \textbf{source metamodel} captures a simplified but representative subset of UML class diagrams:

\begin{lstlisting}[caption={Source UML metamodel (excerpt)},label=lst:uml-mm]
metamodel UMLConcrete {
    enum ClassKind { Entity, Service, ValueObject }
    class Model { name: String }
    class Package { name: String }
    abstract class Classifier { name: String }
    class Class extends Classifier {
        isAbstract: Bool
        isFinal: Bool
        priority: Int
        layerTag: String
        kind: ClassKind
    }
    class Interface extends Classifier { }
    class Enumeration extends Classifier { }
    class DataType extends Classifier { }
    class PrimitiveType extends DataType { }
    abstract class Feature {
        name: String
        visibility: String
        isStatic: Bool
    }
    class Property extends Feature {
        isDerived: Bool
        isReadOnly: Bool
        lower: Int
        upper: Int
    }
    class Operation extends Feature {
        isAbstract: Bool
        isQuery: Bool
    }
    // ... associations: packagedElement, ownedAttribute,
    //     ownedOperation, type, generalization, etc.
}
\end{lstlisting}

The metamodel includes 22 classes with attributes ranging over strings, booleans, integers, and enumerations, connected by associations representing ownership (packages contain classifiers, classifiers own features), typing (properties have types), and inter-classifier relationships (generalization, interface realization, dependency).

The \textbf{target metamodel} represents a Java AST:

\begin{lstlisting}[caption={Target Java AST metamodel (excerpt)},label=lst:java-mm]
metamodel JavaASTConcrete {
    class CompilationUnit { fileName: String }
    class PackageDeclaration { name: String }
    class ClassDeclaration extends TypeDeclaration {
        isAbstract: Bool
        isFinal: Bool
        priority: Int
        layerTag: String
        kind: ClassKind
    }
    class InterfaceDeclaration extends TypeDeclaration { }
    class EnumDeclaration extends TypeDeclaration { }
    class FieldDeclaration extends BodyDeclaration { }
    class MethodDeclaration extends BodyDeclaration {
        isAbstract: Bool
    }
    class ConstructorDeclaration extends BodyDeclaration { }
    // ... TypeReference, Annotation, etc.
}
\end{lstlisting}

\subsection{Transformation Rules}

The transformation is organized into 6 layers, each handling a distinct phase of the translation:

\begin{enumerate}
  \item \textbf{Layer 1 (CreatePackages):} Maps each UML \texttt{Package} to a Java \texttt{PackageDeclaration}.

  \item \textbf{Layer 2 (CreateCompilationUnitsAndTopTypes):} For each packaged classifier (\texttt{Class}, \texttt{Interface}, \texttt{Enumeration}), creates a \texttt{CompilationUnit} containing the appropriate \texttt{TypeDeclaration}. Uses backward links to connect the compilation unit's package reference to the \texttt{PackageDeclaration} created in Layer 1.

  \item \textbf{Layer 3 (CreateOwnedMembers):} Generates \texttt{FieldDeclaration}s from \texttt{Property} elements, \texttt{MethodDeclaration}s from \texttt{Operation} elements, and \texttt{ConstructorDeclaration}s from \texttt{Class} elements.

  \item \textbf{Layer 4 (CompleteSignaturesAndAccessors):} Creates method parameters, return types, constructor parameters, and accessor methods (getters, setters, add/remove for collections).

  \item \textbf{Layer 5 (CreateInterTypeRelations):} Handles generalization (extends), interface realization (implements), dependencies (imports), and property type imports.

  \item \textbf{Layer 6 (AddCommentsAndUtilityMembers):} Generates Javadoc comments and utility methods (equals, hashCode, toString).
\end{enumerate}

A representative rule is shown in Listing~\ref{lst:rule-example}.

\begin{lstlisting}[caption={Example rule: mapping a packaged class to a compilation unit},label=lst:rule-example]
rule Class2CompilationUnit {
    match {
        any pkg : Package
        any cls : Class
        direct pe : packagedElement -- pkg.cls
    }
    apply {
        pkgDecl : PackageDeclaration
        cu : CompilationUnit { fileName = cls.name }
        classDecl : ClassDeclaration {
            name = cls.name,
            visibility = "public",
            isAbstract = cls.isAbstract,
            isFinal = cls.isFinal
        }
        outP : package -- cu.pkgDecl
        outT : types -- cu.classDecl
    }
    backward {
        pkgDecl <--trace-- pkg
    }
}
\end{lstlisting}

This rule demonstrates several key DSLTrans features: (i) \texttt{match} elements with type constraints and direct associations; (ii) \texttt{apply} elements with attribute bindings; (iii) \texttt{backward} links that connect a target element (\texttt{pkgDecl}) to a previously-created element via a trace link from a source element (\texttt{pkg}).

\subsection{Properties}

Properties express correctness requirements as \emph{precondition-postcondition} pairs. A typical property states:

\begin{lstlisting}[caption={Property: every UML package maps to a Java package declaration},label=lst:prop-example]
property PackageHasPackageDeclaration
  "Every UML package maps to a Java package declaration."
{
    precondition {
        any pkg : Package
    }
    postcondition {
        pkgDecl : PackageDeclaration
        pkgDecl <--trace-- pkg
    }
}
\end{lstlisting}

The semantics is: for \emph{every} match of the precondition pattern in the source model, there must \emph{exist} a match of the postcondition pattern in the target model. The trace link \texttt{<--trace--} asserts that the target element was produced from the source element by some rule. Properties can reference multiple source and target elements, include direct associations, and impose attribute constraints.

\begin{figure}[H]
\centering
\includegraphics[width=0.95\textwidth]{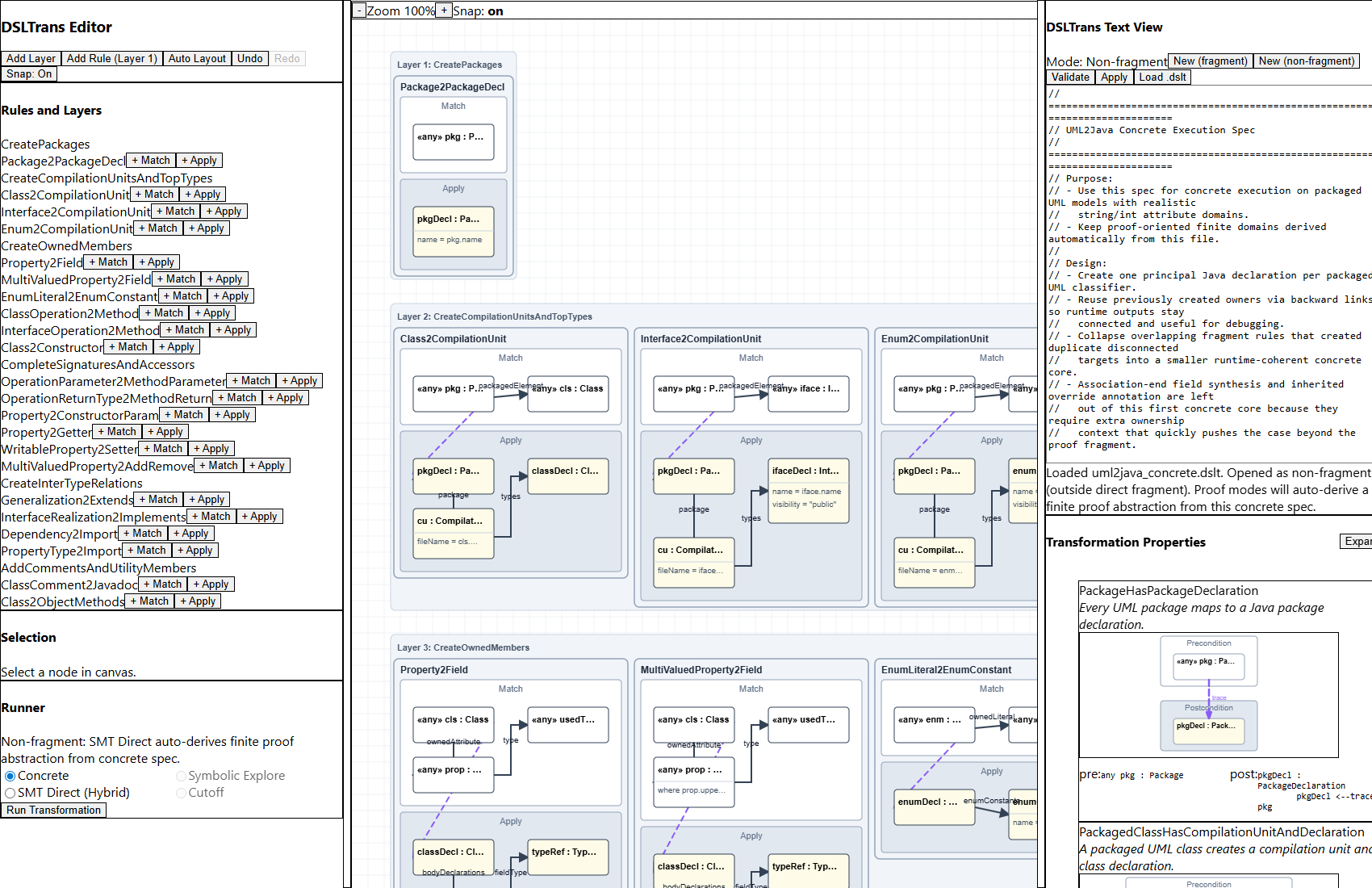}
\caption{UML-to-Java running example in the web IDE. The transformation view exposes the layered DSLTrans rules and the corresponding graphical rule structures used throughout the paper.}
\label{fig:webapp-uml-java-transformation}
\end{figure}

More complex properties involve multiple elements and associations. For example:

\begin{lstlisting}[caption={Property: class-owned property yields an owned field},label=lst:prop-field]
property OwnedPropertyHasOwnedField
  "Each class-owned UML property yields an owned Java field."
{
    precondition {
        any cls : Class
        any prop : Property
        direct owned : ownedAttribute -- cls.prop
    }
    postcondition {
        classDecl : ClassDeclaration
        field : FieldDeclaration
        classDecl <--trace-- cls
        field <--trace-- prop
        outB : bodyDeclarations -- classDecl.field
    }
}
\end{lstlisting}

This property involves two source elements connected by an association, and requires two target elements with a trace link to each source element and a target-side association between them. Note that the postcondition footprint is bounded and purely structural---a characteristic we will exploit in Section~\ref{sec:cutoff}.

\section{The DSLTrans Fragment for Verification}
\label{sec:dsltrans}

We now formalize the fragment of DSLTrans and the class of properties for which our verification approach is complete.

\subsection{Foundational Definitions}

\begin{definition}[Typed Graph / Model]
\label{def:typed-graph}
A \emph{typed graph} over a metamodel $MM = (C, A)$ (where $C$ is a set of classes and $A$ is a set of associations) is a tuple $M = (V, \tau, E, \alpha)$ where:
\begin{itemize}
  \item $V$ is a finite set of vertices (elements),
  \item $\tau: V \to C$ assigns a type (class) to each vertex,
  \item $E \subseteq V \times A \times V$ is a set of labeled edges (links); we write $(v, a, w) \in E$ for an edge from $v$ to $w$ with association label $a \in A$,
  \item $\alpha$ assigns attribute values to vertices.
\end{itemize}
The \emph{size} of the model is $|M| = |V|$.
\end{definition}

\begin{definition}[DSLTrans Rule]
\label{def:rule}
A \emph{rule} $R = (\mathit{Match}, \mathit{Apply}, \mathit{Backward})$ consists of:
\begin{itemize}
  \item A \emph{match pattern} $\mathit{Match}$: a typed graph fragment to be found in the source model. The \emph{match arity} $\mathit{ar}(R) = |V_{\mathit{Match}}|$ is the number of match elements.
  \item An \emph{apply pattern} $\mathit{Apply}$: a typed graph fragment to be created in the target model upon a successful firing. Apply elements may have attribute bindings derived from matched source elements.
  \item A set of \emph{backward links} $\mathit{Backward} \subseteq V_{\mathit{Apply}} \times V_{\mathit{Match}}$ (written with \texttt{<{-}-trace-{-}>} in concrete syntax): each pair $(v_a, v_m)$ declares that the apply element $v_a$ must be connected to an existing target element that was previously created from the source element bound to $v_m$, as recorded by trace links from earlier layers.
\end{itemize}
\end{definition}

\begin{definition}[DSLTrans Transformation]
\label{def:transformation}
A \emph{transformation} $T = (L_1, L_2, \ldots, L_n)$ is an ordered sequence of layers. Each layer $L_i$ contains a finite set of rules $\{R_{i,1}, \ldots, R_{i,k_i}\}$. Execution proceeds layer by layer. The output of layer $i$ is visible to backward links of rules in layers $j > i$, but not to rules in the same layer. Within a layer, rules execute independently and cannot interfere with each other.
\end{definition}

\begin{definition}[Rule Firing]
\label{def:firing}
A rule $R \in L_i$ \emph{fires} on a source model $M$ with accumulated target state $\hat{T}_{<i}$ (the target model produced by layers $1, \ldots, i{-}1$) when both of the following conditions hold:
\begin{enumerate}
  \item \textbf{Source match.} There exists an injective, type-compatible function $\mu: V_{\mathit{Match}} \to V_M$ that respects all structural (association) constraints in $\mathit{Match}$ and satisfies all attribute guards.
  \item \textbf{Backward resolution.} For each backward link $(v_a, v_m) \in \mathit{Backward}$, there exists a unique target element $t \in V_{\hat{T}_{<i}}$ such that $\tau(t)$ is compatible with $\tau(v_a)$ and the trace record of $\hat{T}_{<i}$ contains a link from the source element $\mu(v_m)$ to $t$.
\end{enumerate}
When $R$ fires, its apply pattern is instantiated: fresh target elements are created for each $v_a \notin \mathrm{dom}(\mathit{Backward})$, backward-linked apply elements are bound to their resolved targets, attribute bindings are evaluated, and trace links from source elements to newly created target elements are recorded.
\end{definition}

\paragraph{Notation.}
We write $\mathit{Fires}(R, M, \hat{T}_{<i})$ for the set of valid firing bindings of rule $R \in L_i$ on source model $M$ with target state $\hat{T}_{<i}$, i.e., the pairs $(\mu, \rho)$ satisfying conditions 1 and 2 above, where $\rho$ maps backward-linked apply elements to their resolved targets.
When backward links are absent ($\mathit{Backward} = \emptyset$), a firing reduces to a source-side match, and we write $\mathit{Matches}(R, M)$ for the set of injective functions $\mu$ satisfying condition~1.
We write $\mathrm{MtTypes}(R) := \{\tau(v) \mid v \in V_{\mathit{Match}}\} \subseteq C$ for the set of source \emph{class types} that appear on vertices of $R$'s match pattern.

\subsection{The F-LNR Fragment}
\label{subsec:flnr}

Our verification approach targets the \emph{Local, Non-Recursive} (F-LNR) fragment of DSLTrans.

\begin{definition}[F-LNR Fragment]
\label{def:flnr}
A transformation $T$ belongs to the F-LNR fragment if it satisfies:
\begin{description}
  \item[R1 (No indirect links):] Only direct associations are allowed in match patterns. Indirect links (transitive closure) are excluded.
  \item[R2 (Bounded match size):] There exists $m \in \mathbb{N}$ such that $\mathit{ar}(R) \le m$ for all rules $R \in T$.
  \item[R3 (Layered execution):] Rules are organized into ordered layers, and backward links point only to earlier layers. This is a structural property of DSLTrans and is guaranteed by construction.
  \item[R4 (Acyclic backward dependencies):] The dependency graph induced by backward links is acyclic. Since backward links always reference earlier layers, this follows from R3.
  \item[R5 (Finite attribute domains):] All attributes have finite domains: Bool, bounded-range Int, finite-vocabulary String, or Enum.
  \item[R6 (Monotonicity):] Rules only \emph{add} elements, links, and trace links to the target model; they never delete or modify existing elements. This is implicit in DSLTrans semantics.
\end{description}
\end{definition}

The rationale for each restriction is as follows:
\begin{itemize}
  \item \textbf{R1} ensures that rule matching is purely local, avoiding the need to compute transitive closures that could create unbounded dependencies.
  \item \textbf{R2} bounds the ``blast radius'' of any single rule match, ensuring that adding one rule to the witness model adds at most $m$ elements.
  \item \textbf{R5} ensures that the attribute space is finite, making SMT encoding decidable without resorting to undecidable string or arithmetic theories.
  \item \textbf{R6} (monotonicity) is the cornerstone of the cutoff argument: because rules only add structure, reducing the source model can only reduce the target model, preserving violations of positive properties.
\end{itemize}

\subsection{Properties and the G-BPP Fragment}
\label{subsec:gbpp}

\begin{definition}[DSLTrans Property]
\label{def:property}
A \emph{property} $P = (\mathit{Pre}, \mathit{Post}, \mathit{Trace})$ consists of:
\begin{itemize}
  \item A \emph{precondition pattern} $\mathit{Pre}$: a typed graph fragment over the source metamodel.
  \item A \emph{postcondition pattern} $\mathit{Post}$: a typed graph fragment over the target metamodel.
  \item A set of \emph{trace constraints} $\mathit{Trace} \subseteq V_{\mathit{Post}} \times V_{\mathit{Pre}}$: each pair $(v_t, v_s)$ asserts that the target element matching $v_t$ was produced from the source element matching $v_s$, as recorded by a trace link.
\end{itemize}
A transformation $T$ \emph{satisfies} $P$ on source model $M$, written $T(M) \models P$, if and only if: for every injective, type-compatible match $\mu_{\mathit{pre}}: V_{\mathit{Pre}} \to V_M$ that respects the structural constraints of $\mathit{Pre}$, there exists an injective, type-compatible match $\mu_{\mathit{post}}: V_{\mathit{Post}} \to V_{T(M)}$ that respects the structural constraints of $\mathit{Post}$ and satisfies every trace constraint---i.e., for each $(v_t, v_s) \in \mathit{Trace}$, the trace record of $T(M)$ contains a link from $\mu_{\mathit{pre}}(v_s)$ to $\mu_{\mathit{post}}(v_t)$.
We write $T(M) \not\models P$ when some precondition match has no corresponding postcondition witness.
\end{definition}

The cutoff theorem (Section~\ref{sec:cutoff}) applies to the following restricted fragment of properties.

\begin{definition}[G-BPP Property Fragment]
\label{def:gbpp}
A property $P = (\mathit{Pre}, \mathit{Post}, \mathit{Trace})$ belongs to the \emph{Bounded Positive Pattern} (G-BPP) fragment if:
\begin{description}
  \item[P1 (Bounded pattern size):] There exists $p \in \mathbb{N}$ such that $|V_{\mathit{Pre}}| \le p$ and $|V_{\mathit{Post}}| \le p$.
  \item[P2 (No indirect links):] No transitive-closure associations in $\mathit{Pre}$ or $\mathit{Post}$.
  \item[P3 (Local trace requirements):] Trace constraints relate only elements explicitly named in $\mathit{Pre}$ and $\mathit{Post}$. The property cannot assert ``some unknown source element produced this target.''
  \item[P4 (Positive existence/traceability semantics):] Satisfaction is monotone with respect to adding target structure. The property requires the \emph{existence} of a bounded postcondition witness for each precondition match; it does not express global absence, uniqueness, or exclusivity constraints.
\end{description}
\end{definition}

The key insight is that G-BPP properties have a bounded ``footprint'': the postcondition witness involves at most $p$ target elements, $p$ source elements, and finitely many local trace and structural links. This bounded footprint is what makes the cutoff argument possible.

\subsection{Concrete Execution versus Proof-Time Abstraction}
\label{subsec:abstraction}

In practice, transformations are written with unrestricted attribute domains (arbitrary strings, unbounded integers) for realistic execution, whereas Condition R5 requires finite domains for the proof. We therefore formalize an \emph{attribute abstraction} from concrete attribute valuations to finite proof-time valuations.

\begin{definition}[Attribute Abstraction]
\label{def:abstraction}
An \emph{attribute abstraction} is an attribute-only abstraction consisting of a family of total functions on source and target attribute valuations. For each source class $C$ and target class $D$, let
\[
\alpha_S^C : \mathrm{Val}^{\mathit{conc}}_C \to \mathrm{Val}^{\mathit{abs}}_C
\qquad\text{and}\qquad
\alpha_T^D : \mathrm{Val}^{\mathit{conc}}_D \to \mathrm{Val}^{\mathit{abs}}_D,
\]
where $\mathrm{Val}^{\mathit{conc}}_C$ and $\mathrm{Val}^{\mathit{abs}}_C$ (respectively $\mathrm{Val}^{\mathit{conc}}_D$ and $\mathrm{Val}^{\mathit{abs}}_D$) are the concrete and abstract attribute-valuation spaces for that class. These functions lift pointwise to models by leaving vertices, edges, types, rules, layers, and trace structure unchanged, and replacing only attribute valuations by their abstract images. The abstraction must satisfy:
\begin{enumerate}
  \item \textbf{Attribute-only scope:} the structural part of the specification is preserved (classes, associations, rules, layers, and trace structure unchanged).
  \item Abstract attribute domains are finite.
  \item Rule compatibility is preserved: for the abstracted attributes, rule applicability depends only on distinctions preserved by the abstraction.
  \item Property observables are preserved: attributes referenced by the property remain distinguishable.
\end{enumerate}
\end{definition}

\begin{lemma}[Claim Boundary]
\label{lem:claim-boundary}
A proof on the abstract specification is absolute for the abstract model family. It lifts to the concrete transformation only for correctness claims whose truth is preserved by $(\alpha_S, \alpha_T)$.
\end{lemma}

\begin{proof}[Proof sketch]
The abstract specification preserves all structural and control-flow distinctions (by conditions 1, 3 of Definition~\ref{def:abstraction}), so any proof on it is exact for the abstract model family. Lifting to the concrete transformation requires only that the claim not distinguish concrete valuations that collapse to the same abstract valuation, which holds when condition 4 (observable preservation) covers the claim's referenced attributes.
\end{proof}

In practice, structural and traceability properties transfer directly, while exact attribute-value properties transfer only if the abstraction preserves the relevant distinctions.

\section{The Cutoff Theorem}
\label{sec:cutoff}

The cutoff theorem is the mathematical foundation of our verification approach. It establishes that for F-LNR transformations and G-BPP properties, if a property violation exists for \emph{any} input model (of any size), then a violation exists for a \emph{small} input model whose size is bounded by a computable constant $K$.

\subsection{Theorem Statement and Proof Sketch}

\begin{theorem}[Cutoff Theorem for F-LNR $\times$ G-BPP]
\label{thm:cutoff}
Let $T \in \text{F-LNR}$ be a transformation and let $P = (\mathit{Pre}, \mathit{Post}, \mathit{Trace}) \in \text{G-BPP}$ be a property---with precondition pattern $\mathit{Pre}$, postcondition pattern $\mathit{Post}$, and trace constraints $\mathit{Trace}$ as in Definition~\ref{def:gbpp}. Let:
\begin{itemize}
  \item $c$ = number of source classes relevant to $P$ (reachable in the dependency graph),
  \item $m$ = maximum match size over relevant rules,
  \item $p$ = maximum pattern size of $P$ (i.e., $\max(|V_{\mathit{Pre}}|, |V_{\mathit{Post}}|)$),
  \item $d$ = depth of the backward dependency graph (set $d' = \max(d, 1)$),
  \item $a$ = effective association arity for transitive well-formedness closure, i.e., the maximum number of additional elements forced by the least mandatory-association closure of any retained element,
  \item $r$ = total number of relevant rules for $P$.
\end{itemize}
Although Definition~\ref{def:gbpp}(P1) is stated with an existentially quantified bound, we instantiate $p$ as the \emph{exact} value $\max(|V_{\mathit{Pre}}|, |V_{\mathit{Post}}|)$ for the given $P$---the minimal natural number satisfying P1---and the cutoff formulas below use this $p$ (not a strictly larger surrogate).
Define the cutoff bound $K = \min(K_{\text{coarse}}, K_{\text{sharp}}, K_{\text{tight}})$ where:
\begin{align}
  K_{\text{coarse}} &= c \cdot (m + p) \cdot d' \cdot (a + 1) \label{eq:k-coarse} \\
  K_{\text{sharp}} &= p \cdot (1 + m \cdot r) \cdot d' \cdot (a + 1) \label{eq:k-sharp} \\
  K_{\text{tight}} &= p \cdot (1 + (m-1) \cdot r \cdot d) \cdot (a + 1) \label{eq:k-tight}
\end{align}
If there exists a model $M$ such that $T(M) \not\models P$, then there exists a model $M'$ with $|M'| \le K$ such that $T(M') \not\models P$.
\end{theorem}

\begin{proof}[Proof sketch]
Because the property is local (bounded pattern of size $p$) and the rules are local (bounded match of size $m$, no transitive closure), any property violation is caused by a small cluster of elements. Starting from an arbitrary counterexample $M$, we construct a bounded witness submodel $M'$ in four stages.

\begin{enumerate}
  \item \textbf{Violation witness.} Since $T(M) \not\models P$, there exists an injective, type-compatible match $\mu_{\mathit{pre}}: V_{\mathit{Pre}} \to V_M$ of the precondition pattern in $M$ with no corresponding postcondition witness. Define $\mathit{Support}_0 = \mu_{\mathit{pre}}(V_{\mathit{Pre}})$, with $|\mathit{Support}_0| \le p$ (since $\mu_{\mathit{pre}}$ is injective).

  \item \textbf{Bounded rule participation.} Iteratively expand the support to include match elements of relevant rules that fire on the current support. Each rule match adds at most $m - 1$ new elements (since it must share at least one element with the existing support to be connected). After $d$ layers of backward dependency traversal, the support stabilizes at $\mathit{Support}_d$ with:
  \begin{equation}
    |\mathit{Support}_d| \le p \cdot (1 + (m-1) \cdot r \cdot d)
  \end{equation}
  Informally, this is a sound over-approximation: we start from at most $p$ seed elements, and across at most $d$ dependency layers we conservatively allow participation from any of the $r$ relevant rules at each layer, each contributing at most $m-1$ fresh connected elements. This yields the conservative multiplier $1 + (m-1)rd$ on the initial $p$ elements.

  \item \textbf{Association closure.} To maintain well-formedness, add the least mandatory-association closure of the support. By definition of $a$, each retained element forces at most $a$ additional elements under this transitive closure, so:
  \begin{equation}
    |\mathit{Closure}(\mathit{Support}_d)| \le |\mathit{Support}_d| \cdot (a + 1)
  \end{equation}

  \item \textbf{Witness preservation via monotonicity.} Let $M'$ be the submodel of $M$ induced by $\mathit{Closure}(\mathit{Support}_d)$. The logical shape of the argument is: (a) the property requires the \emph{existence} of a postcondition witness (positive existence, P4); (b) target growth is monotone in source size---adding source elements can only produce \emph{more} target structure, never less (R6); (c) therefore, removing irrelevant source structure from $M$ to obtain $M'$ can only shrink the target model, and cannot create the missing postcondition witness.

  Formally, we show $T(M') \subseteq T(M)$ as graphs by induction on layers. Consider layer $L_i$ and suppose the target states agree up to this point: $\hat{T}'_{<i} \subseteq \hat{T}_{<i}$. For any firing $(\mu, \rho) \in \mathit{Fires}(R, M', \hat{T}'_{<i})$: (i)~the source match $\mu$ maps into $M' \subseteq M$, so $\mu$ is also a valid source match in~$M$; (ii)~each backward resolution $\rho(v_a) = t$ requires a trace link from $\mu(v_m)$ to $t$ in $\hat{T}'_{<i}$, and since $\hat{T}'_{<i} \subseteq \hat{T}_{<i}$, the same trace link exists in $\hat{T}_{<i}$, so backward resolution succeeds in $M$ as well. Thus every firing in the smaller execution is also a firing in the larger one, target elements and trace links produced by $M'$ are a subset of those produced by $M$, and the inductive step holds. Since the postcondition witness was absent from the larger $T(M)$, it remains absent from the smaller $T(M')$, and the violation persists.
\end{enumerate}

Each of $K_{\text{coarse}}$, $K_{\text{sharp}}$, and $K_{\text{tight}}$ is obtained by a different sound over-approximation of stage 2. The bound $K_{\text{tight}}$ uses the connected-component argument ($m-1$ new elements per step), $K_{\text{sharp}}$ conservatively replaces $m-1$ by $m$, and $K_{\text{coarse}}$ further replaces $r$ with $c \cdot m / p$. After stage 3, each bound is multiplied by $(a+1)$, so all three soundly bound $|M'|$; therefore $|M'| \le K = \min(K_{\text{coarse}}, K_{\text{sharp}}, K_{\text{tight}})$. Appendix~\ref{app:cutoff-derivations} details each derivation.
\end{proof}

\paragraph{Where the argument breaks.}
The monotone witness-preservation argument above relies critically on three properties of the F-LNR fragment. The argument would fail if:
\begin{itemize}
  \item \textbf{Negative application conditions (NACs)} were allowed: a rule conditioned on the \emph{absence} of an element could fire in $M'$ but not in $M$, breaking the subset relationship $T(M') \subseteq T(M)$.
  \item \textbf{Deletion or update semantics} were present: in the larger model $M$, a rule might delete a target element; shrinking to $M'$ could cause that rule not to fire, allowing the element to survive in $T(M')$ and potentially completing a postcondition witness that did not exist in $T(M)$.
  \item \textbf{Transitive or nonlocal matching} were used: rule matches could then depend on arbitrarily long chains of elements, making the bounded support construction invalid.
  \item \textbf{Properties required global absence or uniqueness}: a property such as ``there is no element of type $T$'' is \emph{anti-monotone}---shrinking the source reduces what the transformation can build, so a violation present in $T(M)$ can \emph{disappear} in $T(M')$ when the smaller input simply produces fewer $T$-instances. The same holds for uniqueness or exclusivity constraints. In all these cases the small-model witness construction does not apply, because shrinking the model makes the property \emph{easier}, not harder, to satisfy.
\end{itemize}

\subsection{Significance}
\label{subsec:cutoff-significance}

The cutoff theorem transforms bounded model checking from a \emph{heuristic} (``we checked up to size $k$ and found nothing'') into a \emph{complete verification method} (``we checked up to the theorem-valid bound $K$ and can therefore guarantee the property holds for all models of any size''). The bound $K$ is a computable constant that depends only on the static characteristics of $(T, P)$---it has no dependence on the size of any particular input model.

\subsection{Claim Boundary}
\label{subsec:claim-boundary}

To help the reader calibrate the scope of the result, we state explicitly what is and is not established by the cutoff theorem.

\begin{description}
  \item[What is proved:] Completeness of bounded verification for the F-LNR fragment (Definition~\ref{def:flnr}) and G-BPP properties (Definition~\ref{def:gbpp}). If a property violation exists for any input model, it can be found within the computable bound $K$.

  \item[What is not claimed:] The theorem does not cover arbitrary DSLTrans transformations (e.g., those with indirect links or unbounded recursion), negative or global properties (absence, uniqueness, exclusivity), or unrestricted attribute domains. Properties outside G-BPP---such as those requiring global absence of a pattern or counting constraints---require separate justification.

  \item[What lifts to the concrete setting:] When verification uses an abstracted proof specification (Section~\ref{subsec:abstraction}), a proof result is absolute for the abstract model family. It transfers to the concrete transformation only for claims whose truth is preserved by the attribute abstraction map (Lemma~\ref{lem:claim-boundary}). Structural and traceability properties transfer directly; exact attribute-value claims transfer only if the abstraction preserves the relevant distinctions.
\end{description}

\section{Soundness-Preserving Optimizations}
\label{sec:optimizations}

While the cutoff theorem guarantees a finite search space, the \emph{raw} global bound $K$ computed from the theorem formulas alone can be very large---hundreds or even thousands for complex properties. However, raw $K$ is a poor predictor of actual verification difficulty: the optimizations below interact to reduce the \emph{effective} search space by orders of magnitude (e.g., from a raw $K = 735$ to an effective per-class bound of 6; see Table~\ref{tab:ablation}). What matters for tractability is the optimized bound vector, not the raw cutoff. We develop a suite of optimizations that reduce this effective search space while preserving the soundness and completeness of the verification. Crucially, these optimizations \emph{compose safely}: each independently preserves the verification guarantee, and their combination maintains it.

\paragraph{SMT encoding strategy.}
All optimizations in this section operate on top of a common encoding that discharges the bounded verification problem to an off-the-shelf SMT solver. Following the strategy outlined in Section~\ref{sec:introduction}, our prover targets the Z3 SMT solver~\cite{demoura2008z3} as its backend, and we use ``Z3'' and ``the SMT solver'' interchangeably in what follows. Given a transformation $T$, a bound $K$, and a property $P = (\mathit{Pre}, \mathit{Post})$, the prover asks the solver to decide
\begin{equation}
  \Phi \;=\; \Phi_{\mathit{world}} \;\wedge\; \Phi_{\mathit{rules}} \;\wedge\; \Phi_{\mathit{pre}} \;\wedge\; \neg\Phi_{\mathit{post}},
  \label{eq:smt-core}
\end{equation}
where $\Phi_{\mathit{world}}$ encodes a well-formed source model of size $\le K$, $\Phi_{\mathit{rules}}$ encodes the full firing semantics of every rule---source match satisfaction, backward link resolution against trace links from earlier layers, and target element creation (Definition~\ref{def:firing}), $\Phi_{\mathit{pre}}$ asserts that the precondition pattern is matched, and $\Phi_{\mathit{post}}$ encodes the existence of a postcondition witness including its target elements, target-side links, attribute constraints, and the required trace links back to the chosen precondition elements. Thus $\neg\Phi_{\mathit{post}}$ asserts that no such witness exists. If the solver returns $\solverUNSAT$, no violation exists within the bound---and by the cutoff theorem the property $\resHolds$ for \emph{all} input models. If it returns $\solverSAT$, the satisfying assignment encodes a concrete counterexample. The full construction of each sub-formula is given in Appendix~\ref{app:smt-encoding}; the optimizations below reduce the size of~$\Phi$ while preserving this equivalence.

\subsection{Fragment-Based Verification}
\label{subsec:fragments}

Rather than encoding the entire transformation, we verify each property against a \emph{fragment} $F \subseteq T$ containing only the rules relevant to that property. This reduces both the number of SMT variables and the cutoff bound $K$ (since $r$ and $d$ are computed on the fragment).

\begin{lemma}[Monotonic Fragment Soundness]
\label{lem:fragment-soundness}
Let $F \subseteq T$ be a subset of layers. For any source model $M$, if $F(M) \models P$, then $T(M) \models P$.
\end{lemma}

\begin{proof}[Proof sketch]
By R6 (monotonicity), $F(M) \subseteq T(M)$ as graphs: every element, link, and trace created by $F$ is also present in $T(M)$. Since $P$ is a positive existence property (P4), if a postcondition witness exists in $F(M)$, the same witness exists in $T(M)$.
\end{proof}

This means that if the SMT solver proves $P$ holds ($\solverUNSAT$) on the fragment, the result is absolute for the full transformation. A $\resViolated$ result on the fragment, however, might be \emph{spurious}: rules in $T \setminus F$ might produce the missing postcondition structure.

\subsection{CEGAR Refinement}
\label{subsec:cegar}

We handle spurious counterexamples via Counterexample-Guided Abstraction Refinement (CEGAR)~\cite{clarke2003cegar}.

\begin{lemma}[Witness Completeness via CEGAR]
\label{lem:cegar}
Let $M$ be a counterexample such that $F(M) \not\models P$. If for all rules $R \in T \setminus F$, no firing is possible---i.e., $\mathit{Matches}(R, M) = \emptyset$ (no valid source match exists)---then $T(M) \not\models P$.
\end{lemma}

\begin{proof}[Proof sketch]
An omitted rule can contribute target structure only by firing on $M$, which requires at minimum a valid source match (Definition~\ref{def:firing}, condition~1). If no omitted rule has a source match, no omitted rule fires, so $T(M)$ and $F(M)$ produce identical target structure, and the violation on $F(M)$ persists in $T(M)$.
\end{proof}

If any omitted rule does match, its layer is added to the fragment and verification is retried. The implementation realises this idea as a \emph{progressive, counterexample-driven refinement loop} rather than an immediate jump to the full rule set: (i) try the \emph{minimal fragment} (fewest layers covering the property's postcondition types and trace requirements); (ii) on $\resViolated$, enter a CEGAR loop that, at each iteration, inspects the concrete counterexample~$M$ and adds back \emph{only} those omitted relevant rules whose match pattern is satisfiable on the source side of $M$ (Lemma~\ref{lem:cegar}'s match witness). The cutoff $K$ is recomputed for the enlarged fragment and, if it remains within the configured per-property budget, verification is retried; (iii) only if this loop reaches the full set of relevant rules (the \emph{baseline fragment}) without resolving the verdict, or if the refined cutoff exceeds the budget, the prover falls back to the baseline fragment as a final attempt. This preserves Lemma~\ref{lem:cegar}'s witness-completeness guarantee while keeping the SMT bound small for the typical case in which the violation is explained by a few additional rules rather than the entire transformation.

\subsection{Per-Class Bound Reduction}
\label{subsec:per-class}

The global bound $K$ allocates the same number of potential elements for every class, wasting capacity on classes barely touched by the property. We derive per-class bounds through a sequence of lemmas.

\begin{lemma}[Witness Restriction]
\label{lem:witness-restriction}
For F-LNR $\times$ G-BPP, the source side of any counterexample can be reduced to only the precondition witness elements plus the least mandatory-association closure required for well-formedness.
\end{lemma}

\begin{proof}[Proof sketch]
The violation is triggered by the bounded precondition witness (G-BPP). By F-LNR locality, only source elements in the witness or forced by mandatory associations participate in rule matches relevant to the violation. Removing all other source elements can only eliminate unrelated rule firings; by monotonicity (R6), the violation persists.
\end{proof}

\begin{corollary}[Source Seed Sufficiency]
\label{cor:source-seed}
The source-side per-class bounds satisfy: $K_S^{\text{base}} = p_S$ (the number of precondition elements of class $S$), extended by mandatory association closure.
\end{corollary}

\begin{proof}[Proof sketch]
Immediate from Lemma~\ref{lem:witness-restriction}: the base count for each source class $S$ is exactly $p_S$ (the precondition witness elements of that class), with additional instances arising only from mandatory-association closure.
\end{proof}

\begin{lemma}[Target Production Bound]
\label{lem:target-production}
Given source bounds $K_S$ for each source class $S$ (for example, the source-side bounds obtained from Corollary~\ref{cor:source-seed} after mandatory closure), the number of target elements of class $C$ is bounded by:
\begin{equation}
  K_C^{\text{base}} = p_C + \sum_{R \in \mathit{Relevant}} N_R \cdot m_C(R)
\end{equation}
Here $K_C^{\text{base}}$ means the \emph{base} bound for target class $C$: it counts the target instances of class $C$ that may be needed \emph{before} adding any further target elements forced only by target-side mandatory-association closure. The term $p_C$ is the number of postcondition elements of class $C$ that the property witness may require directly. The summation then adds the extra class-$C$ target instances that relevant transformation rules may produce:
\begin{itemize}
  \item for each relevant rule $R$, $N_R$ bounds how many times $R$ may fire on the bounded source witness;
  \item $m_C(R)$ is how many apply elements of class $C$ a single firing of $R$ can create;
  \item therefore $N_R \cdot m_C(R)$ bounds the total number of class-$C$ target elements contributed by rule $R$.
\end{itemize}
The firing bound $N_R$ is estimated by
\begin{equation}
  N_R \le \prod_{S \in \mathrm{MtTypes}(R)} K_S^{n_S(R)}
\end{equation}
which bounds the number of possible rule firings. Here $\mathrm{MtTypes}(R)$ is the set of source classes used in the match of rule $R$, and $n_S(R)$ is the number of match elements of class $S$ in $R$. Intuitively, for each source class used by $R$, we multiply by the number of available bounded choices for that class, once for each match position of that class. For example, if $R$ matches two elements of class $A$ and one of class $B$, then $N_R \le K_A^2 K_B$.
This explicitly covers properties whose postcondition witness requires several target elements produced by multiple firings of the \emph{same} rule on different bounded source bindings.
\end{lemma}

\begin{proof}[Proof sketch]
Target elements of class $C$ arise from postcondition requirements ($p_C$) or from rule firings. Rule $R$ can fire at most $N_R \le \prod_{S \in \mathrm{MtTypes}(R)} K_S^{n_S(R)}$ times (one firing per injective binding over bounded source counts), each creating at most $m_C(R)$ elements of class $C$. Summing over all relevant rules gives the stated bound.
\end{proof}

\begin{theorem}[Per-Class Bound Theorem]
\label{thm:per-class}
A counterexample exists with each class count $\le K_C$, where the per-class bounds $K_C$ are given by the least fixed point of source seeds, source mandatory closure, target production bounds, target mandatory closure, and a global cap at $K$.
\end{theorem}

\begin{proof}[Proof sketch]
The preceding lemmas enumerate every source of per-class demand: source seeds, mandatory closure, and target production. Starting from seed counts and iterating these obligations produces a monotone sequence of per-class bounds, capped at $K$, that stabilizes at a least fixed point. Any minimal witness respects each obligation, so its per-class counts are bounded componentwise by this fixed point. See Appendix~\ref{app:per-class-proof} for the full argument.
\end{proof}

The practical impact is substantial. As a concrete instance, consider the UML-to-Java property \texttt{PropertyHasField}---``every UML \texttt{Property} traces to a Java \texttt{FieldDeclaration}''---whose full derivation is carried out in Appendix~\ref{app:worked-example}. The theorem's global bound is $K = 78$, but per-class analysis determines that only 1 \texttt{Property} slot and 1 \texttt{Class} slot are needed on the source side, reducing the effective search space from 78 uniform slots to a handful of per-class slots.

\subsection{Trace-Aware Dependency Pruning}
\label{subsec:trace-aware}

The cutoff parameters $r$ (relevant rules) and $d$ (dependency depth) critically affect the bound $K$. A naive computation includes every rule that can produce any target element of the required type. We refine this using trace and attribute awareness.

In what follows it is convenient to think of a \emph{relevance mode} as any concrete procedure that, given a property and a rule set, returns values for $r$ and $d$. The prover ships with three such modes---\emph{legacy}, \emph{trace-aware}, and \emph{trace-and-attribute-aware}---introduced in the paragraphs below. We write $r_{\text{mode}}$ and $d_{\text{mode}}$ for the values returned by an arbitrary such mode and $K_{\text{mode}}$ for the bound the Cutoff Theorem then yields from them; stronger modes prune more aggressively and so produce smaller (tighter) values. The next lemma states the soundness condition that every such mode must satisfy.

\begin{lemma}[Soundness of a Relevance Mode]
\label{lem:dependency-pruning}
Fix a property $\varphi$, and let $r_{\text{true}}$ and $d_{\text{true}}$ denote, respectively, the minimal relevant-rule count and minimal dependency depth that suffice for the Cutoff Theorem's witness construction. If a relevance mode returns values $r_{\text{mode}}$ and $d_{\text{mode}}$ with
\[
  r_{\text{true}} \le r_{\text{mode}} \quad \text{and} \quad d_{\text{true}} \le d_{\text{mode}},
\]
then the resulting bound $K_{\text{mode}}$ also satisfies the conditions of the Cutoff Theorem; that is, $K_{\text{mode}}$ is a sound (possibly looser) witness bound.
\end{lemma}

\begin{proof}[Proof sketch]
Each cutoff formula is monotone in $r$ and $d$: replacing the minimal values $r_{\text{true}}, d_{\text{true}}$ by any over-approximations $r_{\text{mode}}, d_{\text{mode}}$ can only enlarge the resulting bound, so $K_{\text{mode}} \ge K_{\text{true}}$. Any witness fitting under $K_{\text{true}}$ therefore also fits under $K_{\text{mode}}$; soundness is preserved, though tightness may be lost.
\end{proof}

The three modes mentioned above instantiate this scheme, ranging from a coarse baseline to the refinement actually used in the theory.
\paragraph{Legacy baseline.}
The coarsest safe approximation is \emph{legacy} relevance: if the property requires a target element of type $T$, then every rule that can produce an element of type $T$ is considered relevant. This is sound because it certainly includes every rule that might participate in a witness, but it is often very imprecise: many such rules produce the right target type for the \emph{wrong} source provenance.

\paragraph{Trace-aware refinement.}
The key observation is that a postcondition trace requirement does not ask merely for ``some target element of type $T$''; it asks for a target element of type $T$ that was produced \emph{from the appropriate source element}. Thus, if the property contains a trace requirement from a source type $S$ to a target type $T$, a witness can only use rules that match an $S$-typed source element and create a fresh $T$-typed target element. This replaces the coarse question
\[
\text{``which rules can create } T\text{?''}
\]
by the more informative question
\[
\text{``which rules can create } T \text{ traced from } S\text{?''}
\]
and therefore prunes away rules that have the right target type but the wrong provenance.

The same idea applies recursively to dependencies created by backward links. If a later rule needs an existing target element of type $T$ traced from a match element of type $S$, then only earlier rules that can produce that same $(S,T)$ trace pair need to appear in the dependency closure. Consequently, trace-aware relevance and dependency depth remain sound over-approximations, but are typically much smaller than the legacy ones.

\paragraph{Implementation refinement.}
The prover also implements a further \emph{trace-and-attribute-aware} pruning mode, which discards trace-compatible rules whose attribute bindings are incompatible with the property's attribute requirements. This is useful in practice, but the theoretical development in this section only requires the trace-aware refinement above.

For UML-to-Java, trace-aware mode reduces $r$ from 10+ to 2--3 for many properties, shrinking $K$ by factors of 3--5$\times$.

\subsection{Equisatisfiable SMT Encoding}
\label{subsec:equisatisfiable}

\begin{lemma}[Factored Postcondition Equivalence]
\label{lem:factored-post}
If the postcondition graph has independent connected components $C_1, \ldots, C_k$, the monolithic existential quantification is equivalent to the factored formula:
\begin{equation}
  \Phi_{\text{factored}} = \bigwedge_{i=1}^k \left(\exists \vec{e} \in C_i : \mathit{links}(\vec{e})\right)
\end{equation}
reducing the SMT search space from $O(N^{|P|})$ to $O(|P| \cdot N^{\max|C_i|})$, where $N$ is the total number of target slots and $|P|$ abbreviates the number of postcondition elements $|V_{\mathit{Post}}|$.
\end{lemma}

\begin{proof}[Proof sketch]
Independent components share no existential variables, so the existential quantifiers distribute across the conjunction: satisfying the full postcondition is equivalent to satisfying each component separately. In the overall SMT formula of Appendix~\ref{app:smt-encoding}, this factored form rewrites the $\neg \Phi_{\mathit{post}}$ component while leaving $\Phi_{\mathit{world}}$, $\Phi_{\mathit{rules}}$, and $\Phi_{\mathit{pre}}$ unchanged. The complexity improvement follows because the solver solves one subproblem per component, dominated by the largest.
\end{proof}

Additional equisatisfiable optimizations include:
\begin{itemize}
  \item \textbf{Property slicing:} Remove classes and associations not referenced by the property or its relevant rules.
  \item \textbf{Metamodel slicing:} Remove unreachable metamodel elements.
  \item \textbf{Attribute domain reduction:} Collapse unused attribute domains to singletons.
\end{itemize}

\subsection{Composability}

\begin{theorem}[Composability of Optimizations]
\label{thm:composability}
Any finite composition of the above optimizations, in any order, preserves the soundness and completeness of the verification.
\end{theorem}

\begin{proof}[Proof sketch]
We reason by preservation invariants on the \emph{current transformed verification problem}. Each optimization belongs to one of three classes. First, witness-preserving refinements shrink the searched rule/model space but preserve the existence of a cutoff-valid witness whenever one exists for the current problem; this includes the bound refinements and trace-aware dependency pruning, and fragment selection once paired with its CEGAR repair condition. Second, equisatisfiable rewrites transform the SMT formula while preserving its SAT/\solverUNSAT status; this includes factored postconditions and the slicing/domain-reduction rewrites. Third, exact refinement steps such as CEGAR may reject a spurious fragment counterexample, but only by restoring omitted structure and re-running verification, so they preserve final verdict correctness.

These invariants are closed under composition. Composing witness-preserving refinements yields another witness-preserving refinement, because a witness surviving the later refinement also survives the earlier safe over-approximation chain. Composing equisatisfiable rewrites yields another equisatisfiable rewrite by transitivity of logical equivalence. Interleaving exact refinement steps with either class cannot invalidate the established invariant, because each refinement step re-establishes a sound current problem before verification continues. Therefore, after any finite sequence and any ordering of the optimizations, the final verification problem still has the same \resHolds/\resViolated answer as the original theorem-complete procedure.
\end{proof}

\section{The DSLTrans Prover}
\label{sec:prover}

We now describe the architecture, implementation, and empirical evaluation of the prover. The implementation, the web-based IDE, the corpus of transformations and properties used in the evaluation, and the supporting evaluation artifacts are all distributed together as a public companion repository~\cite{lucio2026companion}.

\subsection{Architecture}

The prover consists of three main components:

\begin{enumerate}
  \item \textbf{Parser and model layer}: Parses \texttt{.dslt} specification files into an in-memory AST, supporting metamodels with inheritance, associations with multiplicities, transformation rules with all F-LNR constructs, and properties with preconditions, postconditions, and backward links.

  \item \textbf{Cutoff analysis}: Implements fragment validation, relevant rule computation with legacy/trace-aware/attribute-aware modes, cutoff bound computation ($K_{\text{coarse}}$, $K_{\text{sharp}}$, $K_{\text{tight}}$), per-class bound fixed-point analysis, and minimal/baseline fragment construction.

  \item \textbf{SMT-direct verifier}: Encodes the bounded verification problem as a single Z3 formula using three sub-encoders:
  \begin{itemize}
    \item \textbf{Model encoder}: Creates bounded source and target worlds with per-class slots, existence variables, association matrices, and attribute variables with finite domains.
    \item \textbf{Rule encoder}: Encodes rule firing semantics---for each injective match binding, a firing variable implies source match satisfaction, backward link resolution, and target element creation with ``at most one creator'' constraints.
    \item \textbf{Property encoder}: Encodes the violation condition $\mathit{Pre} \wedge \neg \mathit{Post}$, using factored postconditions and optional incremental checking.
  \end{itemize}
\end{enumerate}

Supporting components include:
\begin{itemize}
  \item \textbf{Concrete execution engine}: Executes transformations on concrete XMI/Ecore models for testing and cross-validation.
  \item \textbf{Abstraction pipeline}: Automatically synthesizes proof specifications from concrete specifications.
  \item \textbf{Inheritance flattening}: When properties use abstract types, generates concrete-class variants for verification and lifts counterexamples back.
\end{itemize}

The abstraction pipeline is automated and property-directed. It scans rule \texttt{where} clauses, rule guards, apply bindings, and property constraints to identify attribute observables, then finite-izes only the attribute domains while preserving the structural transformation and property vocabulary. It also propagates domains through direct attribute-copy bindings, so the resulting proof specification is synthesized generically from the given transformation and properties.

Operationally, the prover validates each synthesized abstraction before using it for proof: for every rule guard and match-element \texttt{where} clause, it checks that the abstract domains preserve the same Boolean control-flow distinctions as the concrete domains for the supported predicate fragment used by the prover. The resulting proof specification is then re-checked against the verifiable fragment conditions before SMT solving. In the web IDE, this enables fully automatic proof attempts via abstraction before SMT verification.

\paragraph{The concrete-to-proof pipeline.}
The end-to-end verification workflow proceeds in four steps:
\begin{enumerate}
  \item \textbf{Author a concrete specification}: The transformation is written with unrestricted attribute domains (arbitrary strings, unbounded integers) for realistic execution and testing.
  \item \textbf{Synthesize a proof specification}: The abstraction pipeline automatically derives a proof-oriented specification with finite attribute domains, preserving structural and control-flow distinctions.
  \item \textbf{Verify properties}: The prover runs bounded model checking on the proof specification, using the cutoff theorem and the optimization stack.
  \item \textbf{Interpret results}: Proof results are absolute for the abstract model family and transfer to the concrete transformation for claims whose truth is preserved by the abstraction (Lemma~\ref{lem:claim-boundary}).
\end{enumerate}

\begin{figure}[H]
\centering
\resizebox{\textwidth}{!}{%
\begin{tikzpicture}[
  >=Latex,
  every node/.style={align=center, font=\small},
  artifact/.style={draw, thick, rounded corners=8pt, fill=teal!12,
                   minimum width=6cm, minimum height=1.2cm},
  core/.style={draw, line width=1.6pt, fill=blue!8,
               minimum width=6cm, minimum height=1.35cm},
  util/.style={draw, thick, dashed, fill=gray!10,
               minimum width=4.6cm, minimum height=1.2cm},
  flow/.style={->, line width=1.5pt},
  aux/.style={->, thick, dashed}
]

\node[artifact] (spec) at (0, 0)
  {\textbf{.dslt Specification}\\transformation + properties};

\node[core] (parser) at (0, -2.2)
  {\textbf{Parser \& Model Layer}};

\node[core] (abstraction) at (0, -4.6)
  {\textbf{Abstraction Pipeline}\\{\footnotesize concrete-to-proof synthesis}};

\node[core] (cutoff) at (0, -7.0)
  {\textbf{Cutoff Analysis}\\{\footnotesize relevance, $K$, per-class bounds}};

\node[core] (smt) at (0, -9.4)
  {\textbf{SMT-Direct Verifier}\\{\footnotesize model/rule/property encoders}};

\node[artifact] (result) at (0, -11.6)
  {\textbf{Verification Result}\\$\resHolds$ / $\resViolated$ / counterexample};

\draw[flow] (spec) -- (parser);
\draw[flow] (parser) -- (abstraction);
\draw[flow] (abstraction) -- (cutoff);
\draw[flow] (cutoff) -- (smt);
\draw[flow] (smt) -- (result);

\node[util] (runtime) at (-8.0, -4.6)
  {\textbf{Concrete Executor}\\XMI/Ecore execution};

\node[util] (flattening) at (8.0, -9.4)
  {\textbf{Inheritance Flattening}\\concrete-class expansion};

\draw[aux] (runtime.east) -- (parser.west);
\draw[aux] (runtime.east) -- (abstraction.west);
\draw[aux] (flattening.west) -- (smt.east);

\node[artifact, minimum width=2.2cm, minimum height=0.65cm, font=\scriptsize]
  (leg1) at (8.0, -12.6) {data artifact};
\node[core, minimum width=2.2cm, minimum height=0.65cm, font=\scriptsize]
  (leg2) at (8.0, -13.4) {core component};
\node[util, minimum width=2.2cm, minimum height=0.65cm, font=\scriptsize]
  (leg3) at (8.0, -14.2) {support component};

\end{tikzpicture}%
}
\caption{Architecture of the DSLTrans prover. Solid arrows show the main verification pipeline from concrete specification through proof-specification synthesis to SMT checking; dashed arrows show supporting service interactions.}
\label{fig:prover-architecture}
\end{figure}

\subsection{The Web-Based IDE}

The prover is integrated into a web-based development environment (\emph{DSLTrans Browser Studio}) built with React and Vite. The IDE supports:

\begin{itemize}
  \item \textbf{Textual editing}: Full-featured editor for \texttt{.dslt} specifications with syntax awareness.
  \item \textbf{Visual rendering}: Graphical views of metamodels, transformation rules, and property patterns using Konva.
  \item \textbf{Concrete execution}: Run transformations on XMI input models and inspect outputs.
  \item \textbf{Formal verification}: Run SMT-direct verification with streaming per-property progress, displaying results ($\resHolds$/$\resViolated$/$\resUnknown$), timing, and counterexample information.
  \item \textbf{Automatic abstraction}: For non-fragment specifications, the IDE automatically synthesizes a proof specification and runs verification on it.
\end{itemize}

The backend bridges to the Python prover via a Node.js server that spawns Python processes, enabling real-time streaming of verification results to the browser.

\begin{figure}[H]
\centering
\begin{minipage}[t]{0.49\textwidth}
  \centering
  \includegraphics[width=\textwidth]{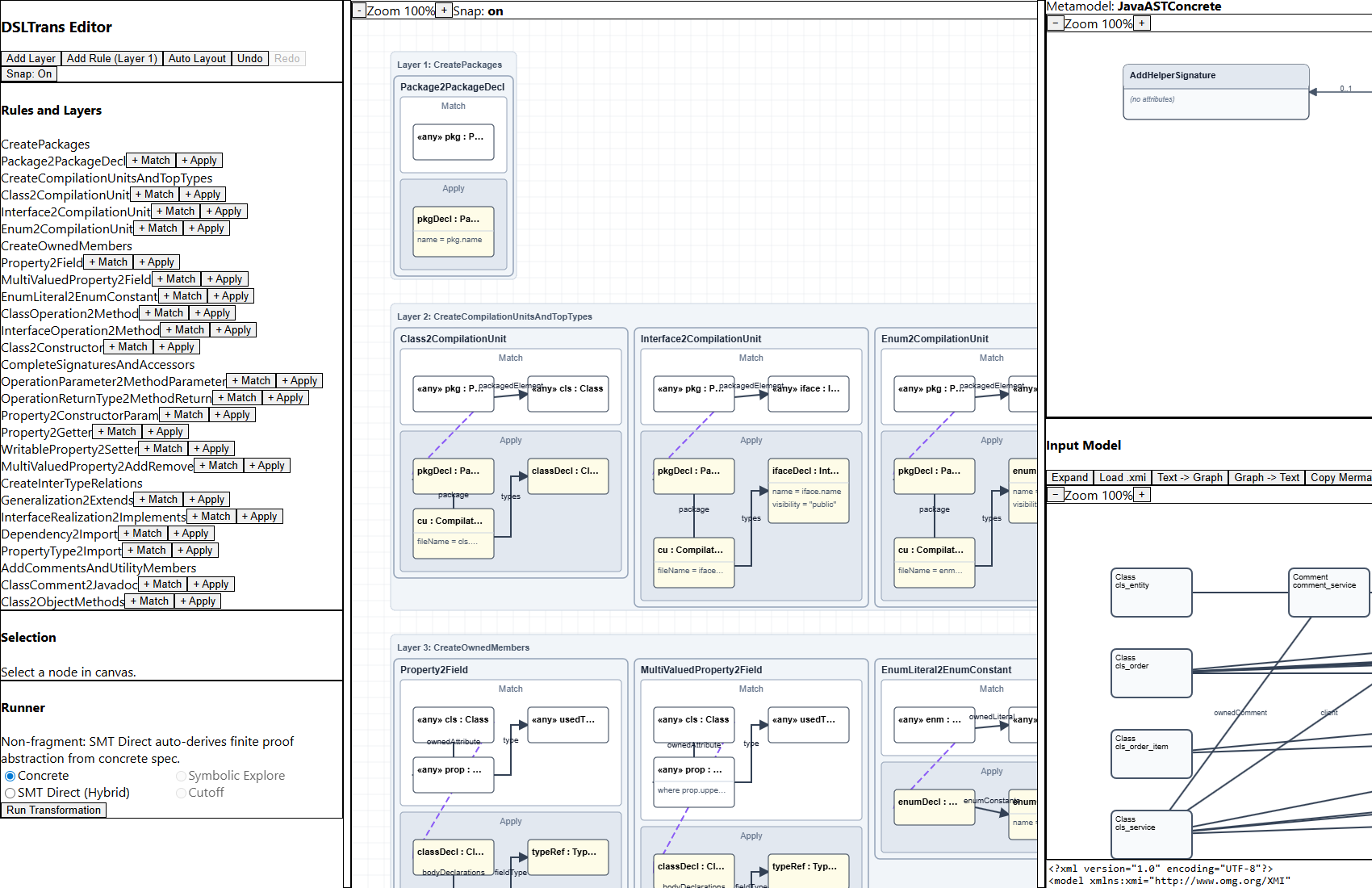}
  \\[-0.3em]
  {\small (a) Model-loading and context panel}
\end{minipage}\hfill
\begin{minipage}[t]{0.27\textwidth}
  \centering
  \includegraphics[width=\textwidth]{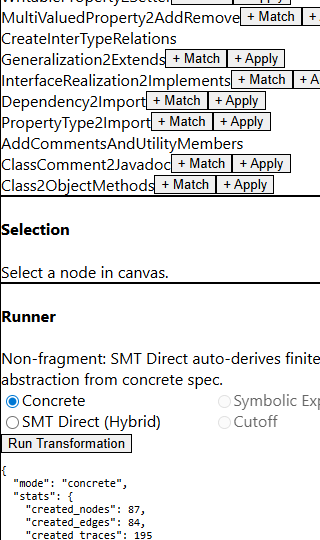}
  \\[-0.3em]
  {\small (b) Concrete execution result}
\end{minipage}\hfill
\begin{minipage}[t]{0.20\textwidth}
  \centering
  \includegraphics[width=\textwidth]{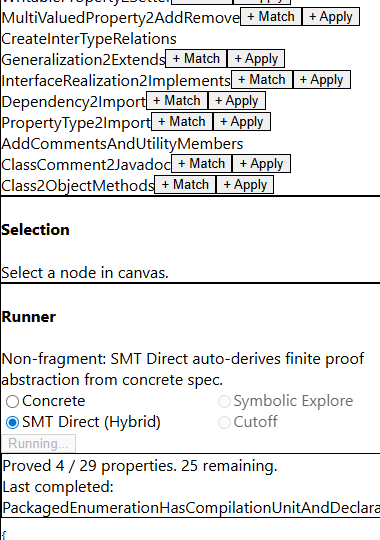}
  \\[-0.3em]
  {\small (c) Proof mode result}
\end{minipage}
\caption{The DSLTrans Browser Studio workflow for the UML-to-Java example. The IDE supports loading source models, executing the transformation concretely, and running SMT-based verification from the same environment.}
\label{fig:webapp-workflow}
\end{figure}

\subsection{Experimental Setup}
\label{subsec:experimental-setup}

All experiments were run on a single machine with an Intel Core i7-12700H (14 cores), 32\,GB RAM, running Windows 10. The prover is implemented in Python 3.12 with Z3 4.13 as the SMT backend. Each property is verified independently with a wall-clock timeout of 600\,s (10\,min); properties exceeding this budget are reported as \resUnknown{}. The dependency mode is trace-and-attribute-aware (Section~\ref{subsec:trace-aware}), and per-class bounds (Theorem~\ref{thm:per-class}) are enabled. All timings are single-run wall-clock measurements, not averaged. For concrete transformations, proof specifications are synthesized automatically from the concrete specification via the abstraction pipeline (Section~\ref{subsec:abstraction}) and regenerated for each run.

\begin{table}[H]
\centering
\caption{Concrete UML-to-Java verification results (excerpt).}
\label{tab:uml2java-results}
\scriptsize
\begin{tabularx}{\textwidth}{>{\raggedright\arraybackslash}X r r r l}
\toprule
\textbf{Property} & $K_{\text{base}}$ & $K_{\min}$ & \textbf{Time (s)} & \textbf{Result} \\
\midrule
PackageHasPackage\-Declaration & 2 & 2 & 0.03 & $\resHolds$ \\
OwnedPropertyHas\-OwnedField & 73 & 4 & 29.78 & $\resHolds$ \\
ClassHasConstructor & 7 & 7 & 0.15 & $\resHolds$ \\
DependencyYieldsImport & 58 & 5 & 10.13 & $\resViolated$ \\
OperationParameterHas\-OwnedParameter & 58 & 5 & 0.51 & $\resViolated$ \\
ClassMappedToInterface\-Declaration\_ShouldFail & 30 & 3 & 0.03 & $\resViolated$ \\
MultiValuedPropertyHas\-ArrayFieldType & 405 & 5 & 100.51 & $\resHolds$ \\
MultiValuedPropertyHas\-AddHelper & 735 & 6 & 370.99 & $\resHolds$ \\
\bottomrule
\end{tabularx}
\setlength{\tabcolsep}{6pt}
\renewcommand{\arraystretch}{1}
\end{table}

Three properties in the table are $\resViolated$: \texttt{ClassMappedToInterfaceDeclaration\_ShouldFail} is a negative test, intentionally expected to fail---its counterexample acts as a regression check for a known non-guarantee. The other two (\texttt{DependencyYieldsImport} and \texttt{OperationParameterHas\-OwnedParameter}) are \emph{boundary properties}: semantically meaningful correctness claims that are stronger than what the current transformation rules guarantee. Their counterexamples document the boundary of the proven correctness envelope and guide future strengthening of the transformation (see Section~\ref{subsec:interpreting-results}).

Table~\ref{tab:corpus-summary} summarizes the full corpus results. The 29 concrete transformations are organized into primary benchmarks (individually listed), a graph-mapping suite of seven structurally similar transforms, and synthetic/stress benchmarks. The three transformations grouped as ``K-boundary witnesses'' in the synthetic block are the specs used for the empirical cutoff-bound validation experiment described in Appendix~\ref{app:cutoff-validation}; their per-class bounds are exercised at $K{-}3$ through $K{+}3$ to corroborate that the theorem-derived bound is sufficient and that no single per-class slot can be removed without changing the verdict, while leaving open whether a strictly smaller theoretical bound might also be sound.

\begin{table}[H]
\centering
\caption{Verification results across the full concrete benchmark corpus. The graph-mapping suite contains seven structurally similar transformations sharing a common rule template; individual suite entries appear in Appendix~\ref{app:corpus-details}.}
\label{tab:corpus-summary}
\scriptsize
\begin{tabularx}{\textwidth}{>{\raggedright\arraybackslash}X r r r r r r r}
\toprule
\textbf{Transformation} & \textbf{Rules} & \textbf{Layers} & \textbf{Props} & \textbf{Holds} & \textbf{Viol.} & \textbf{Unk.} & \textbf{Max (s)} \\
\midrule
\multicolumn{8}{l}{\textit{Primary benchmark suites}} \\
ATLCompiler & 55 & 9 & 20 & 15 & 5 & 0 & 4.45 \\
BibTeX2DocBook & 11 & 3 & 33 & 28 & 5 & 0 & 0.36 \\
BPMN2Petri & 12 & 3 & 5 & 5 & 0 & 0 & 17.64 \\
Class2Relational & 18 & 6 & 29 & 23 & 6 & 0 & 1.22 \\
Ecore2JsonSchema & 5 & 3 & 20 & 15 & 5 & 0 & 3.89 \\
FSM2PetriNet & 3 & 3 & 10 & 5 & 5 & 0 & 2.70 \\
KM32OWL & 6 & 3 & 20 & 15 & 5 & 0 & 3.92 \\
OCLCompiler & 29 & 6 & 40 & 32 & 6 & 2 & 70.00 \\
Persons & 9 & 3 & 9 & 8 & 1 & 0 & 1.86 \\
PNML2NUPN & 5 & 3 & 20 & 15 & 5 & 0 & 3.81 \\
SimplePDL2Tina & 6 & 3 & 20 & 15 & 5 & 0 & 3.85 \\
Statechart2Flow & 3 & 3 & 20 & 10 & 10 & 0 & 0.12 \\
Tree2Graph & 2 & 2 & 6 & 3 & 3 & 0 & 1.52 \\
UML2Java & 22 & 6 & 29 & 23 & 6 & 0 & 472.43 \\
University2ITSystem & 8 & 5 & 21 & 16 & 5 & 0 & 2.34 \\
UseCase2Activity & 3 & 3 & 12 & 8 & 4 & 0 & 0.04 \\
\midrule
\multicolumn{8}{l}{\textit{Graph-mapping suite (7 transforms, 3 rules / 3 layers each)}} \\
\quad Aggregate & 21 & --- & 140 & 77 & 63 & 0 & 4.11 \\
\midrule
\multicolumn{8}{l}{\textit{Synthetic and stress benchmarks}} \\
K-boundary witnesses (3) & --- & --- & 8 & 4 & 4 & 0 & 0.77 \\
Scalability & 67 & 14 & 17 & 15 & 2 & 0 & 3.73 \\
Stress (2) & 278 & --- & 400 & 200 & 200 & 0 & 83.83 \\
\midrule
\textbf{Total (29 transforms)} & & & \textbf{899} & \textbf{552} & \textbf{345} & \textbf{2} & \\
\bottomrule
\end{tabularx}
\end{table}

\subsection{Interpreting the Results}
\label{subsec:interpreting-results}

Three verdicts appear in the results, and their interpretation differs:
\begin{description}
  \item[\resHolds:] The property is proved within the theorem-backed bound $K$. By the cutoff theorem, this is a complete guarantee: the property holds for all input models of any size.
  \item[\resViolated:] The SMT solver found a concrete counterexample---a specific source model and the resulting target model that violates the property. This is useful both as a bug report and as a boundary witness documenting a non-guarantee of the current transformation.
  \item[\resUnknown:] The solver did not terminate within the timeout budget. This is an implementation limitation, not a refutation of the theoretical framework: the property may well hold or fail, but the current solver budget was insufficient to decide it.
\end{description}

Properties in the corpus fall into three categories. \emph{Ordinary properties} express expected structural invariants (e.g., ``every UML package maps to a Java package declaration''). \emph{Negative tests}, marked by a dedicated naming suffix indicating that failure is expected, are intentionally expected to be \resViolated{}; their counterexamples serve as regression checks for known non-guarantees. \emph{Boundary properties} are semantically meaningful but stronger than the current transformation contract---they mark behaviors that the transformation does not yet guarantee and document the boundary of the proven correctness envelope.

The concrete corpus demonstrates breadth across transformation sizes, domains, and structural patterns:
\begin{itemize}
  \item \textbf{Standard benchmarks}: Class2Relational and UML2Java are community-standard examples~\cite{atl-zoo,transformation-tool-contest}.
  \item \textbf{Compiler and compiler-like suites}: ATLCompiler (55 rules, the largest non-stress transformation) and OCLCompiler exercise deep rule hierarchies and complex type structures.
  \item \textbf{Behavioral transformations}: UseCase2Activity, Statechart2Flow, FSM2PetriNet, BPMN2Petri, and SimplePDL2Tina cover process and state-machine mappings.
  \item \textbf{Schema and ontology transformations}: Ecore2JsonSchema, KM32OWL, PNML2NUPN, and BibTeX2DocBook cover metamodel translation and document generation.
  \item \textbf{Graph-mapping family}: Seven structurally similar transforms (Component2Deployment, ER2Relational, Family2SocialNetwork, FeatureModel2Configuration, Mindmap2Graph, Organization2AccessControl, Workflow2PetriNet) share a common rule template, providing breadth across domains while testing the prover's consistency on a homogeneous pattern.
  \item \textbf{Stress benchmarks}: Scalability (67 rules, 14 layers), mega-stress (78 rules, 100 properties), and ultra-stress (200 rules, 300 properties) test the prover under extreme workloads.
\end{itemize}

The mixed outcomes are a feature, not a limitation: 552 properties are proved, 345 are falsified with concrete counterexamples (many intentionally through negative tests that are expected to fail), and only 2 remain undecided. This shows that the method both proves nontrivial properties, detects real non-guarantees, and honestly reports its practical limits under current solver budgets.

\paragraph{Sensitivity to abstraction.}
Because verification operates on the abstracted proof specification (Section~\ref{subsec:abstraction}), changes to the abstraction policy can alter cutoff sizes, SMT encodings, and timings. However, such changes do not invalidate the proof framework itself---they merely change the practical workload. Results must therefore be regenerated when the abstraction policy changes, and the timings reported here reflect the current abstraction settings.

\section{Why Is This Tractable?}
\label{sec:tractability}

The tractability of our approach is not a single insight but rather the result of four mutually reinforcing factors. We discuss each in turn.

\subsection{Monotone Semantics of DSLTrans}

The monotone (additive-only) semantics of DSLTrans is the mathematical foundation that enables every other optimization. Without monotonicity:
\begin{itemize}
  \item The cutoff theorem would not hold: a rule that \emph{deletes} elements could invalidate postcondition witnesses in ways that depend on the global model size.
  \item Fragment-based verification would be unsound: adding rules could \emph{remove} previously-satisfied postconditions.
  \item Per-class bounds would be invalid: the witness-restriction lemma relies on the fact that restricting the source model can only shrink the target model.
\end{itemize}

The absence of negative application conditions (NACs) is equally critical: NACs would make rule firing depend on the \emph{absence} of elements, breaking the monotone relationship between source and target model sizes.

\subsection{Leveraging Z3's Scalability}

The move from path condition enumeration~\cite{lucio2010validation,oakes2015fully,oakes2018full} to direct SMT encoding fundamentally changes the verification architecture. Instead of the prover managing combinatorial explosion, we \emph{delegate} it to Z3, which brings decades of engineering in:
\begin{itemize}
  \item SAT solver heuristics (DPLL, CDCL, restart strategies),
  \item Theory solvers (finite domains, bit-vectors, arrays),
  \item Incremental solving and lemma learning,
  \item Formula simplification and preprocessing.
\end{itemize}

The key insight is that the SMT encoding is \emph{structurally regular}: source and target worlds are arrays of bounded slots, rule firings are conjunctions of local constraints, and postconditions are bounded existential queries. This structure is well-suited to Z3's decision procedures.

\subsection{The Cutoff Theorem}

The cutoff theorem converts an infinite verification problem (``for all models of any size'') into a finite one (``for all models up to size $K$''). However, the raw bound $K$ from the theorem formulas is rarely the bottleneck in practice: per-class analysis, fragment selection, and trace-aware pruning interact to replace a single large $K$ with a vector of small, tightly fitted per-class bounds. As Table~\ref{tab:ablation} illustrates, raw bounds in the hundreds or thousands routinely collapse to effective bounds in single digits. It is this optimized bound vector, not the raw cutoff, that determines whether the SMT solver can decide the query in tractable time.

\subsection{Iterative Co-Evolution of Theory and Implementation}
\label{subsec:coevolution}

A distinctive aspect of this work is the \emph{tight feedback loop} between mathematical theory and implementation. Several optimizations that began as ad-hoc code-level improvements were later generalized into the theory, and vice versa:

\begin{itemize}
  \item \textbf{Per-class bounds}: Originally a heuristic ``slot allocation'' trick, later formalized into the Per-Class Bound Theorem with the witness-restriction and target-production lemmas.
  \item \textbf{Trace-aware relevance}: Originally a code-level filter (``only include rules that produce the right trace''), later formalized as a safe over-approximation lemma and integrated into the cutoff bound computation.
  \item \textbf{Factored postconditions}: Originally an engineering optimization, later proved to be an equisatisfiable encoding with the Factored Postcondition Equivalence lemma.
  \item \textbf{Lazy target-world closure}: Target well-formedness constraints (mandatory associations on the target side) were initially encoded eagerly, creating enormous formulas. A lazy CEGAR-style refinement---defer closure constraints until the solver produces a model violating them---was found to be exact and substantially faster. This was later justified by the observation that closure constraints are equisatisfiable under the monotone fragment.
\end{itemize}

This iterative process was greatly accelerated by AI-assisted development, which enabled rapid prototyping of both mathematical proofs and code implementations. The speed of iteration allowed mathematics to serve as an ``abstract reasoning language'' for guiding code development: lemma statements encode architectural decisions, and proofs encode the justification for optimizations.

\subsection{Optimization Impact}
\label{subsec:optimization-impact}

To illustrate the cumulative effect of the optimization stack, Table~\ref{tab:ablation} compares baseline and optimized cutoff bounds for representative UML-to-Java properties. The baseline bound $K_{\text{base}}$ uses the cutoff theorem alone (the tightest of the three formulas); the optimized bound $K_{\min}$ applies per-class analysis, trace-aware dependency pruning, and fragment selection.

\begin{table}[H]
\centering
\caption{Effect of the optimization stack on selected UML-to-Java properties. $K_{\text{base}}$ is the theorem-only bound; $K_{\min}$ is the bound after per-class analysis, trace-aware pruning, and fragment selection.}
\label{tab:ablation}
\small
\begin{tabularx}{\textwidth}{>{\raggedright\arraybackslash}X r r l}
\toprule
\textbf{Property} & $K_{\text{base}}$ & $K_{\min}$ & \textbf{Result} \\
\midrule
PackageHasPackageDeclaration & 2 & 2 & \resHolds \\
ClassHasConstructor & 7 & 7 & \resHolds \\
OwnedPropertyHasOwnedField & 73 & 4 & \resHolds \\
DependencyYieldsImport & 58 & 5 & \resViolated \\
MultiValuedPropertyHasArrayFieldType & 405 & 5 & \resHolds \\
MultiValuedPropertyHasAddHelper & 735 & 6 & \resHolds \\
\bottomrule
\end{tabularx}
\end{table}

For simple properties (e.g., \texttt{PackageHasPackageDeclaration}), the baseline bound is already small and the optimizations have little effect. For structurally complex properties, the reductions are substantial: \texttt{MultiValuedPropertyHasAddHelper} drops from $K = 735$ to $K = 6$, a 122$\times$ reduction in the bound and a corresponding reduction in SMT formula size. Without per-class analysis, this property would be intractable; with it, verification completes in under 7 minutes.

\section{Working Practices and a Tentative Methodology}
\label{sec:methodology}

\subsection{Toward Test-Driven Property Engineering}

We do not claim a validated methodology: the case studies in this paper are too few, and the workflow we followed was iterative and informal rather than prescriptive. What follows is therefore a \emph{tentative} methodology, distilled retrospectively from the practices that proved useful while developing the transformations and property suites in our case studies. We expect it to evolve as the prover is applied at larger scale.

In this spirit, the recurring practices we found most useful were:

\begin{enumerate}
  \item \textbf{Define expected structural invariants}: What relationships should the transformation preserve? Write these as G-BPP properties.
  \item \textbf{Write negative tests}: Define properties that should be $\resViolated$---expressing unintended mappings. These serve as regression tests for the transformation's specificity.
  \item \textbf{Iterate}: If a $\resHolds$ property yields $\resViolated$, either the transformation has a bug or the property is incorrectly specified. If an intentionally failing negative test yields $\resHolds$, the transformation is too permissive.
  \item \textbf{Refine property witnesses}: When the SMT solver times out, analyze whether the postcondition pattern can be decomposed (factored postconditions) or whether the precondition is unnecessarily broad.
\end{enumerate}

In practice, we found a third category useful in addition to invariants and negative tests: \emph{boundary properties}. These are semantically meaningful, often desirable properties that are stronger than the current transformation contract. They are not arbitrary falsehoods like negative tests; rather, they mark behaviors that the transformation does not yet guarantee, either because the current rules are intentionally permissive, because the property had to be weakened to match implemented behavior, or because further instrumentation/refinement of the transformation would be needed to make the property provable. Keeping such properties in the suite helps document the boundary of correctness of the current transformation, guides future strengthening of the rules, and makes the prover useful not only for confirmation but also for property-driven construction and refinement of transformations.

Sub-second verification times for most properties made this style of interactive exploration feasible in our case studies; whether the same workflow scales to larger transformation suites remains to be evaluated.

\subsection{Tractability-Driven Property Design}
\label{subsec:tractability-design}

When a property exceeds the tractability budget, several strategies are available:
\begin{itemize}
  \item \textbf{Precondition strengthening}: Adding attribute guards to the precondition narrows the match space.
  \item \textbf{Precondition specialization}: Replacing abstract/supertype match elements with concrete subtypes reduces per-class slot allocation. The conjunction of specialized variants logically implies the original property.
  \item \textbf{Postcondition decomposition}: Splitting a monolithic postcondition into independent components enables factored encoding.
  \item \textbf{Intrusive markers}: For properties involving polymorphic types (e.g., ``a setter signature exists''), adding marker classes to the target metamodel can disambiguate the postcondition pattern.
  \item \textbf{Property splitting}: Complex properties can be decomposed into simpler sub-properties whose conjunction implies the original.
\end{itemize}

\paragraph{Case study: BPMN2Petri bipartite refinement.}
An initial bipartite-graph diagnostic property for the BPMN-to-Petri-Net transformation used the abstract type \texttt{FlowNode} in both precondition endpoints.  Because \texttt{FlowNode} has 10+ concrete subtypes, the per-class analysis allocated slots for every subtype, yielding $K{=}90$, $\mathit{bound}{=}37$, and a proof time of 383.75\,s.  Worse, the polymorphic typing allowed the solver to construct a \emph{spurious counterexample} in which the same flow node could be matched by multiple mapping rules (which produce Places \emph{and} Transitions), reporting \resViolated{} for a property that should hold.

By specializing the precondition to representative concrete pairs (\texttt{StartEvent} $\to$ \texttt{EndEvent} and \texttt{StartEvent} $\to$ \texttt{ExclusiveGateway}---the pairs where both endpoints map to Places), the refined positive properties \texttt{BipartiteCheck\_StartEnd} and \texttt{BipartiteCheck\_StartGateway} are each verified in 3.4\,s at $\mathit{bound}{=}17$, a \textbf{112$\times$ speedup}.  The result also changes from \resViolated{} to \resHolds{}: the \texttt{SequenceFlow\_PlaceToPlace} rule correctly inserts a silent transition, preserving bipartiteness.  The specialization thus simultaneously improves tractability, eliminates a spurious counterexample, and justifies renaming the refined checks so they are no longer marked as intentionally failing tests.

\paragraph{Case study: UML2Java setter and adder properties.}
In the UML-to-Java transformation, the property \texttt{PropertyHasGetter} was intractable ($>$\,30\,min) because the postcondition pattern---``there exists a method that is a getter for this property''---matched multiple rule outputs indistinguishably.  Adding a \texttt{SetterSignature} marker class to the target metamodel and a corresponding link in the setter-producing rule let the postcondition unambiguously identify the intended getter method, reducing proof time from $>$\,30\,min to $\approx$\,10\,s.  Similarly, the composite property \texttt{MultiValuedPropertyHasAddRemoveHelpers} was split into independent \texttt{HasAddHelper} and \texttt{HasRemoveHelper} sub-properties, each with its own marker.

Table~\ref{tab:tractability-summary} summarizes the property refinement techniques applied across the concrete benchmark corpus and their impact on verification times.

\begin{table}[H]
\centering
\caption{Tractability engineering results. ``Before'' and ``After'' report proof time and verdict for each original/refined property pair.}
\label{tab:tractability-summary}
\scriptsize
\setlength{\tabcolsep}{4pt}
\renewcommand{\arraystretch}{1.05}
\begin{tabularx}{\textwidth}{L{1.8cm}>{\raggedright\arraybackslash}X>{\raggedright\arraybackslash}X r r l l}
\toprule
\textbf{Suite} & \textbf{Property} & \textbf{Technique} & \textbf{Before (s)} & \textbf{After (s)} & \textbf{Before} & \textbf{After} \\
\midrule
BPMN2Petri & Bipartite refinement & Precondition specialization & 383.75 & 3.41 & \resViolated & \resHolds \\
UML2Java & PropertyHasGetter & Intrusive marker & $>$1800 & 10.2 & \resUnknown & \resHolds \\
UML2Java & WritableProperty\-HasSetter & Precondition strengthening & $>$1800 & 9.8 & \resUnknown & \resHolds \\
UML2Java & AddRemoveHelpers & Postcondition decomposition + markers & $>$1800 & 12.1 & \resUnknown & \resHolds \\
OCLCompiler & BuiltInNotHas\-Invoke & Precondition strengthening (attempted) & $>$70 & $>$70 & \resUnknown & \resUnknown \\
\bottomrule
\end{tabularx}
\end{table}

The BPMN2Petri bipartite refinement is notable because it both \emph{improved the result} (eliminating a spurious counterexample caused by polymorphic type ambiguity) and achieved a 112$\times$ speedup; the refined properties were consequently renamed as positive checks. The UML2Java properties demonstrate three complementary techniques---intrusive markers, precondition strengthening, and postcondition decomposition---that together moved three previously-intractable properties from \resUnknown{} to \resHolds{}. The OCLCompiler boundary properties remain intractable due to the deep expression-layer type hierarchy; they mark the current tractability boundary of the approach.

\subsection{AI-Assisted Construction}

AI tools played a significant role in this work:
\begin{itemize}
  \item \textbf{Transformation authoring}: AI was used to translate ATL Zoo transformations into DSLTrans syntax, with manual review for semantic fidelity.
  \item \textbf{Property generation}: AI proposed initial property suites based on transformation structure, which were then refined through iterative verification.
  \item \textbf{Mathematical development}: AI assisted in formulating and refining theorem statements, proof sketches, and the identification of necessary lemmas.
  \item \textbf{Implementation}: The prover code was developed iteratively with AI assistance, enabling rapid exploration of encoding strategies and optimization techniques.
\end{itemize}

The key insight is that AI substantially accelerates the \emph{iteration speed} of the theory-implementation feedback loop described in Section~\ref{subsec:coevolution}. Mathematical reasoning in this context becomes an abstract language for directing implementation: lemma statements encode architectural decisions, and proofs encode the justification for optimizations.

We are aware that AI-assisted theorem statements and AI-assisted prover code lower the cost of producing plausible artifacts but do not by themselves guarantee their correctness. Every theorem statement, lemma, and proof sketch in this paper was reviewed by the author, and the prover was cross-checked empirically through the K-boundary validation experiment of Appendix~\ref{app:cutoff-validation} and through concrete-vs-symbolic agreement on the corpus. We nevertheless regard these checks as complementary to, not a substitute for, the mechanized proof of the cutoff theorem and the systematic mutation testing of the prover that we list as future work in Section~\ref{sec:conclusion}; AI-assisted development is precisely a reason to take those mechanization steps seriously rather than to defer them.

\section{Related Work}
\label{sec:related-work}

We organize related work along three axes: techniques for verifying model transformations (Section~\ref{subsec:rw-mtv}), cutoff results in parameterized verification (Section~\ref{subsec:rw-cutoffs}), and SMT-based bounded verification (Section~\ref{subsec:rw-smt}). The positioning of DSLTrans among model transformation languages, already sketched in the introduction, is revisited briefly in Section~\ref{subsec:rw-languages} from the perspective of verifiability.

\subsection{Verification of Model Transformations}
\label{subsec:rw-mtv}

The verification of model transformations has been studied extensively. Rahim and Whittle~\cite{rahim2015model} provide a survey of the field that gives broader context to the four lines of work below. We categorize the most directly relevant approaches by their underlying verification technique.

\textbf{Theorem proving.} Calegari et al.~\cite{calegari2011model} and Poernomo~\cite{poernomo2008proof} encode transformation properties in interactive theorem provers (Coq, Isabelle). While these approaches offer the strongest guarantees, they require significant manual proof effort and do not scale to realistic transformations without substantial domain-specific automation. Habel and Pennemann~\cite{habel2009correctness} establish foundational results for verifying graph-transformation systems against nested conditions via weakest preconditions; the resulting proof obligations are typically discharged by resolution-style theorem proving rather than SMT, with completeness inherited from the discharger.

\textbf{Operational-semantics-based approaches.} Troya and Vallecillo~\cite{troya2011model} give a rewriting-logic semantics for ATL in Maude, which enables analysis via Maude's reachability and model-checking machinery. This recasts ATL verification as a general-purpose model-checking problem and inherits its scalability profile; it provides no domain-specific completeness argument.

\textbf{Constraint solving and model finders.} A line of work encodes transformation contracts as constraint-satisfaction problems and discharges them with model finders. Anastasakis et al.~\cite{anastasakis2007uml2alloy} translate UML class diagrams and transformation rules into Alloy/SAT. Cabot et al.~\cite{cabot2010verification} verify declarative model-to-model transformations through invariants checked by constraint solvers (UMLtoCSP/EMFtoCSP). B\"uttner et al.~\cite{buttner2012verification} translate ATL transformations into transformation models verified by model finders. Gogolla et al.~\cite{gogolla2014filmstrip} extend the USE tool with filmstrip models for dynamic-property validation. These approaches are the closest neighbors of our SMT encoding: they share the bounded-search engine, but in all of them the search bound is user-supplied and there is no completeness theorem tying the bound to the transformation's structure. Our cutoff theorem mechanically derives a per-class bound from the transformation and proves it sufficient for the F-LNR/G-BPP fragment.

\textbf{Static analysis for ATL.} Cuadrado, Guerra, and de~Lara~\cite{cuadrado2017anatlyzer} perform static analysis of ATL transformations (AnATLyzer) to detect type errors, rule conflicts, and unreachable code. AnATLyzer addresses well-formedness rather than contract-style correctness, but it is the closest practical neighbor in the ATL ecosystem and complements rather than competes with the present work.

\textbf{DSLTrans and related symbolic-execution work.} The original DSLTrans verification approach~\cite{lucio2010validation,oakes2015fully,oakes2018full,oakes2018thesis} used symbolic enumeration of path conditions, with completeness arguments specific to the language but tractability limited by the number of enumerated paths. Oakes' thesis~\cite{oakes2018thesis} also consolidated the theoretical development of this line of work and described supporting tooling for execution and verification. Selim et al.~\cite{selim2013automated} apply a closely related symbolic approach to industrial automotive transformations, encoding ATL transformations as transformation models verified against OCL contracts. Our work replaces path-condition enumeration with SMT-based bounded model checking backed by a cutoff theorem on the per-class slot count, making verification tractable for transformations that were previously out of reach.

\subsection{Cutoff Results in Verification}
\label{subsec:rw-cutoffs}

Cutoff arguments have a long history in parameterized verification. Emerson and Namjoshi~\cite{emerson1995reasoning} established early cutoff results for symmetric ring topologies. Kaiser et al.~\cite{kaiser2010dynamic} extended cutoff detection to programs with dynamic thread creation. The current state of the art across communication primitives, system models, and proof techniques is surveyed by Bloem et al.~\cite{bloem2015decidability}. Our cutoff theorem differs from this body of work in that it is specific to model transformations and exploits structural properties of DSLTrans (layered execution, monotone semantics, bounded match patterns) rather than process symmetry.

\subsection{SMT-Based Verification}
\label{subsec:rw-smt}

Bounded model checking via SAT was introduced by Biere et al.~\cite{biere1999symbolic} for hardware and extended to software via SMT by Clarke et al.~\cite{clarke2004tool}; D'Silva, Kroening, and Weissenbacher~\cite{dsilva2008survey} survey the resulting body of techniques. Bounded SAT/SMT encodings have been applied to model transformations through the model-finder approaches discussed in Section~\ref{subsec:rw-mtv}. The contribution of the present paper is to combine such an SMT encoding with a domain-specific cutoff theorem that turns the bounded check into a complete proof for the F-LNR/G-BPP fragment; we are not aware of a prior result that pairs these two ingredients in this way.

\subsection{Model Transformation Languages}
\label{subsec:rw-languages}

The deliberate Turing-incompleteness of DSLTrans~\cite{barroca2011dsltrans} is the language-design choice that makes the cutoff theorem of this paper possible. Mainstream transformation languages---ATL~\cite{jouault2008atl}, QVT-R~\cite{omg2016qvt}, Epsilon~\cite{kolovos2008epsilon}, Henshin~\cite{arendt2010henshin}---obtain greater expressiveness through helpers, imperative blocks, bidirectionality, recursion, or in-place updates with negative application conditions. Those features are valuable in practice but obstruct the structural arguments (layered execution, monotone semantics, bounded match patterns) on which our completeness result rests. The positive results of this paper should therefore be read as showing what becomes provable when expressiveness is constrained in this specific way, rather than as a critique of the more expressive languages.

\section{Threats to Validity}
\label{sec:threats}

We discuss the main threats to the validity of our results along the standard internal/external/construct dimensions.

\textbf{Internal validity.} The cutoff theorem is stated and proved at the paper level; it is not mechanized in a proof assistant. An error in the theorem or its implementation could therefore make the prover unsound or incomplete without being detected by the tooling. We mitigate this through two complementary checks: (i)~empirical validation of the cutoff bound by re-verifying twelve representative properties under uniform $K\pm1$, $K\pm2$, $K\pm3$ shifts, per-class selective $-1$ perturbations, and concrete witness execution at base$\,-\,1$/base/base$\,+\,1$ support levels (Appendix~\ref{app:cutoff-validation}); and (ii)~cross-validation of symbolic results against the concrete execution engine described in Section~\ref{sec:prover} on representative XMI/Ecore inputs. A mechanized proof of the cutoff theorem and systematic mutation testing of the prover remain future work (Section~\ref{sec:conclusion}).

\textbf{External validity.} The corpus consists of translations of existing transformations into DSLTrans, primarily from the ATL Zoo. The F-LNR fragment excludes important transformation patterns (recursion, NACs, in-place updates), limiting applicability to out-place structural mappings. We also did not perform a head-to-head empirical comparison with neighboring tools such as Alloy-based encodings~\cite{anastasakis2007uml2alloy}, USE/EMFtoCSP~\cite{cabot2010verification,gogolla2014filmstrip}, or AnATLyzer~\cite{cuadrado2017anatlyzer}: each of these tools targets a different language and a different property class, and there is no shared input format on which the same property could be discharged by all of them under comparable assumptions. A faithful comparison would require non-trivial benchmark co-design and is left to future work; the related-work discussion (Section~\ref{sec:related-work}) is therefore positional rather than empirical.

\textbf{Construct validity.} The G-BPP property fragment covers positive existence and traceability properties. Global absence, uniqueness, and non-monotone properties require separate justification.

\section{Conclusion and Future Work}
\label{sec:conclusion}

This paper establishes a complete bounded-verification result for a precise fragment of DSLTrans (F-LNR transformations with G-BPP properties), builds a practical prover architecture around it, and evaluates it on a concrete corpus of 29 transformations with 899 properties spanning compiler lowering, schema translation, behavioral modeling, graph mapping, and synthetic stress tests. The approach rests on three pillars: the monotone semantics of DSLTrans that enables the cutoff theorem; the practical power of Z3 as the verification backend; and a suite of composable, soundness-preserving optimizations that reduce the verification problem to sizes the solver can handle.

The empirical results are encouraging: of 899 properties, 552 are proved, 345 are falsified with concrete counterexamples, and only 2 remain undecided within the timeout budget. The method detects both genuine invariants and real non-guarantees, and tractability-driven property refinement extends its reach to initially intractable cases. Where previous path-condition enumeration timed out on transformations with 15--20 rules, the workflow presented here settles the full corpus---including transformations with up to 55 rules and stress benchmarks with 200 rules and 300 properties---with complete guarantees for decided properties.

Future work includes:
\begin{itemize}
  \item \textbf{Mechanized proof}: Formalizing the cutoff theorem in a proof assistant (Coq or Lean) to eliminate residual trust assumptions.
  \item \textbf{Mutation testing of the prover}: Injecting mutants into the encoder and checking that intended-failure properties are detected, to corroborate the soundness of the implementation.
  \item \textbf{Extended property classes}: Developing cutoff arguments for absence and uniqueness properties under restricted conditions.
  \item \textbf{Fragment extensions}: Carefully extending the F-LNR fragment with limited forms of NACs or bounded recursion while preserving verifiability.
  \item \textbf{Industrial validation}: Applying the approach to safety-critical industrial transformations.
\end{itemize}

\newcommand{\DeferredAppendices}{\clearpage\appendix

\section{SMT Encoding for DSLTrans}
\label{app:smt-encoding}

This appendix describes the SMT encoding used by the prover to translate the bounded verification problem into a Z3 formula.

\subsection{Source and Target Worlds}

For each class $C$ in the source metamodel, the encoder allocates $K_C$ \emph{slots}, each guarded by a Boolean variable $\mathit{exists}_{C,i}$ (for $1 \le i \le K_C$). For each association $A: C_1 \to C_2$, a Boolean matrix $\mathit{link}_{A}[i][j]$ (with $1 \le i \le K_{C_1}$, $1 \le j \le K_{C_2}$) represents whether the link exists. Links can only exist between existing elements:
\begin{equation}
  \mathit{link}_{A}[i][j] \implies \mathit{exists}_{C_1,i} \wedge \mathit{exists}_{C_2,j}
\end{equation}

For each attribute $\mathit{attr}$ of class $C$ with finite domain $D$, an integer variable $\mathit{attr}_{C,i}$ is constrained to $D$. When $C$ does not exist, attributes take a default value.

Multiplicity constraints from the metamodel are encoded as cardinality constraints over link rows/columns.

\subsection{Rule Firing Encoding}

For each rule $R$ with match elements $m_1: S_1, \ldots, m_n: S_n$, the encoder iterates over all injective bindings $\beta = (i_1, \ldots, i_n)$ with $i_k \in \{1, \ldots, K_{S_k}\}$ and $i_j \ne i_k$ when $S_j = S_k$. For each binding, a Boolean variable $\mathit{fires}_{R,\beta}$ is defined:
\begin{equation}
  \mathit{fires}_{R,\beta} \iff \bigwedge_{k} \mathit{exists}_{S_k, i_k} \wedge \bigwedge_{(l,r) \in \mathit{Links}(R)} \mathit{link}_{A_{l,r}}[i_l][i_r] \wedge \mathit{guards}_{R,\beta} \wedge \mathit{backward}_{R,\beta}
\end{equation}

where $\mathit{guards}_{R,\beta}$ encodes attribute constraints (Definition~\ref{def:firing}, condition~1) and $\mathit{backward}_{R,\beta}$ encodes backward link resolution against trace links from earlier layers (Definition~\ref{def:firing}, condition~2).

For each fresh apply element of class $T$ in rule $R$, a choice variable selects which target slot receives the element, with an ``at most one creator'' constraint ensuring no slot is claimed by multiple rule firings.

\subsection{Trace Encoding}

Traces are encoded implicitly through the rule firing structure. When rule $R$ fires with match binding $\beta$ and creates target element $t$ in slot $j$, the trace link $(S_k, i_k) \to (T, j)$ is recorded. Backward link resolution in later layers queries these trace links to identify previously-created target elements.

\subsection{Property Violation Encoding}

The property $P = (\mathit{Pre}, \mathit{Post})$ is encoded as a satisfiability query:
\begin{equation}
  \Phi = \underbrace{\Phi_{\mathit{world}}}_{\text{well-formed models}} \wedge \underbrace{\Phi_{\mathit{rules}}}_{\text{rule semantics}} \wedge \underbrace{\Phi_{\mathit{pre}}}_{\text{precondition holds}} \wedge \underbrace{\neg \Phi_{\mathit{post}}}_{\text{postcondition fails}}
\end{equation}

Here $\Phi_{\mathit{post}}$ requires a full postcondition witness: the chosen target elements exist, all required target-side links hold, any postcondition attribute constraints are satisfied, and every declared property trace link connects each postcondition element to the corresponding precondition element. If Z3 returns $\solverUNSAT$, no violation exists (the property $\resHolds$). If Z3 returns $\solverSAT$, the satisfying assignment encodes a counterexample: a specific source model and the resulting target model that violates the property.

The postcondition encoding uses the factored form when the postcondition graph has independent components, and falls back to monolithic encoding otherwise. For incremental checking, each precondition binding is checked separately with early termination on the first violation found.

\section{Cutoff Bound Derivations}
\label{app:cutoff-derivations}

\subsection{\texorpdfstring{Derivation of $K_{\text{tight}}$}{Derivation of K-tight}}

\begin{enumerate}
  \item \textbf{Initial support.} The violation starts with a match of the precondition pattern. The precondition has at most $p$ elements, so $|\mathit{Support}_0| \le p$. These are the ``seed'' elements.

  \item \textbf{Backward dependency chains.} To construct the minimal witness, we must ensure that the target elements required by the property are actually created. If a rule $R_i$ requires a target element $t$ via a backward link, we trace backward to the specific rule $R_{i-1}$ that created $t$ and include the source elements that $R_{i-1}$ matched. This forms a backward dependency chain. In the worst case, satisfying the property may require tracing through all $r$ relevant rules, so there are at most $r$ chains per seed element.

  \item \textbf{Chain depth.} Because the transformation layers form a strict ordering (R3), the backward dependency graph is a DAG. The maximum length of any chain is the depth $d$ of this DAG.

  \item \textbf{Cost per step (connected-component argument).} At each step backward in a chain, we add the match elements of one rule. A relevant rule matches at most $m$ elements. Because the chain is connected---the rule must share at least one element with the existing support---each step adds at most $m - 1$ \emph{new} elements. When $m = 1$, the matched element is already in the support, contributing 0 new elements; thus $(m-1)$ correctly bounds the additions.

  \item \textbf{Total additions.} Starting with $p$ seeds, tracing at most $r$ chains per seed, each of depth at most $d$, each step adding at most $(m-1)$ elements:
  \[
    \text{additions} \le p \cdot r \cdot d \cdot (m - 1)
  \]
  The $p$ multiplier is conservative: each seed element might independently require its own full set of $r$ chains.

  \item \textbf{Support size.}
  \[
    |\mathit{Support}_d| \le p + p \cdot r \cdot (m-1) \cdot d = p \cdot (1 + (m-1) \cdot r \cdot d)
  \]

  \item \textbf{Association closure.} To ensure the extracted elements form a well-formed model, we include the least transitive mandatory-association closure required by the metamodel (e.g., containment parents and any mandatory ancestors they in turn force). By definition of $a$, each retained element forces at most $a$ additional elements under this closure, so:
  \[
    |\mathit{Closure}(\mathit{Support}_d)| \le |\mathit{Support}_d| \cdot (a + 1)
  \]

  \item \textbf{Final formula.}
  \[
    K_{\text{tight}} = p \cdot (1 + (m-1) \cdot r \cdot d) \cdot (a + 1)
  \]
\end{enumerate}

\subsection{\texorpdfstring{Derivation of $K_{\text{sharp}}$}{Derivation of K-sharp}}

$K_{\text{sharp}}$ is a slightly looser but structurally simpler bound, derived with more conservative approximations.

\begin{enumerate}
  \item \textbf{Initial support.} As before, $|\mathit{Support}_0| \le p$.

  \item \textbf{Rule expansion without connected-component optimization.} Instead of tracking that each rule match shares at least one element with the existing support ($m-1$ new elements), we conservatively assume each match may add up to $m$ elements.

  \item \textbf{Support size.} Tracing backward through $d$ layers, with up to $r$ relevant rules per seed element:
  \[
    |\mathit{Support}_d| \le p + p \cdot r \cdot m \cdot d = p \cdot (1 + m \cdot r \cdot d)
  \]

  \item \textbf{Factoring out $d'$.} Define $d' = \max(d, 1)$. Since $1 \le d'$, we can safely over-approximate the constant $1$ as $d'$:
  \[
    p \cdot (1 + m \cdot r \cdot d) \le p \cdot (d' + m \cdot r \cdot d') = p \cdot (1 + m \cdot r) \cdot d'
  \]
  This factoring also handles the edge case $d = 0$ (a single-layer transformation with no backward dependencies), where $K_{\text{sharp}}$ correctly reduces to $p \cdot (1 + m \cdot r) \cdot (a + 1)$.

  \item \textbf{Association closure and final formula.}
  \[
    K_{\text{sharp}} = p \cdot (1 + m \cdot r) \cdot d' \cdot (a + 1)
  \]
\end{enumerate}

\paragraph{Linear vs.\ exponential growth.}
A natural concern is whether the support might grow exponentially with depth. It does not: the witness construction is \emph{selective}---at each chain step we add only the elements needed to justify one backward dependency, not all possible rule matches. The contribution of rule expansion is therefore \emph{additive} in chain length ($d$), yielding linear growth in the term $p \cdot r \cdot d \cdot (m-1)$ rather than multiplicative growth such as $(1 + r(m-1))^d$. While exponential behavior can appear in the SMT solver's search (as a function of the final bound $K$), the theorem-side increase in $K$ itself is linear in $d$.

\subsection{\texorpdfstring{Derivation of $K_{\text{coarse}}$}{Derivation of K-coarse}}

$K_{\text{coarse}}$ replaces the per-rule counting with a class-based approximation, avoiding dependence on $r$ entirely.

\begin{enumerate}
  \item \textbf{Initial support.} As before, $|\mathit{Support}_0| \le p$.

  \item \textbf{Rule contribution via class counting.} Rather than tracking individual rules and their chain structure, observe that each relevant rule match involves at most $m$ source elements drawn from $c$ distinct source classes. In the worst case, every combination of seed elements with rule matches contributes new elements, bounded by the product $c \cdot m$. Adding the $p$ seed elements:
  \[
    |\mathit{Support}_d| \le c \cdot m + p \le c \cdot (m + p)
  \]
  This replaces the $p \cdot (1 + m \cdot r)$ factor with $c \cdot (m + p)$, which is tighter when $c$ is small relative to $p \cdot r$.

  \item \textbf{Depth and association closure.}
  \[
    K_{\text{coarse}} = c \cdot (m + p) \cdot d' \cdot (a + 1)
  \]
\end{enumerate}

$K_{\text{coarse}}$ tends to dominate (i.e., be the minimum) when the number of relevant source classes $c$ is small, even if $r$ (the number of relevant rules) is large.

\subsection{Worked Example: PropertyHasField}
\label{app:worked-example}

To illustrate how the bounds are computed in practice, we trace through the \texttt{PropertyHasField} property of the UML-to-Java transformation. This property states: ``Every UML \texttt{Property} traces to a Java \texttt{FieldDeclaration}.'' The property has a single precondition element (\texttt{any prop : Property}) and a single postcondition element (\texttt{field : FieldDeclaration}) linked by a trace constraint.

\paragraph{Parameters.}
The cutoff analysis (with trace-aware dependency mode) determines:
\begin{itemize}
  \item $p = 1$: $\max(|V_{\mathit{Pre}}|, |V_{\mathit{Post}}|) = \max(1, 1) = 1$.
  \item $m = 3$: the largest relevant rule matches 3 source elements.
  \item $r = 8$: eight rules are relevant (they produce \texttt{FieldDeclaration} elements with the correct trace type, including rules for static, read-only, collection, and primitive-type properties).
  \item $d = 1$: backward dependency depth is 1 (single-step backward dependency; $d' = \max(d, 1) = 1$).
  \item $a = 5$: transitive mandatory-association closure adds up to 5 additional forced elements per retained element in the UML metamodel.
  \item $c = 5$: five source classes are reachable in the dependency graph.
\end{itemize}

\paragraph{Bound computation.}
\begin{align*}
  K_{\text{coarse}} &= c \cdot (m + p) \cdot d' \cdot (a + 1) = 5 \cdot (3 + 1) \cdot 1 \cdot 6 = 120 \\
  K_{\text{sharp}} &= p \cdot (1 + m \cdot r) \cdot d' \cdot (a + 1) = 1 \cdot (1 + 3 \cdot 8) \cdot 1 \cdot 6 = 150 \\
  K_{\text{tight}} &= p \cdot (1 + (m-1) \cdot r \cdot d) \cdot (a + 1) = 1 \cdot (1 + 2 \cdot 8 \cdot 1) \cdot 6 = 102
\end{align*}

The final bound is $K = \min(120, 150, 102) = 102$. Here $K_{\text{tight}}$ wins because the connected-component argument ($m-1 = 2$ instead of $m = 3$) yields a tighter inner term than $K_{\text{sharp}}$, and $K_{\text{coarse}}$ benefits from having fewer relevant classes ($c = 5$) but is still larger.

\paragraph{Per-class reduction.}
The global bound $K = 102$ would allocate 102 slots for every class. Per-class analysis (Theorem~\ref{thm:per-class}) determines much tighter bounds. On the source side:
\begin{itemize}
  \item $K_{\text{Property}} = 2$ (one seed from the precondition, plus one additional from rule expansion).
  \item $K_{\text{Class}} = 1$, $K_{\text{Package}} = 1$ (mandatory closure from \texttt{Property}).
  \item All other source classes: $K_C \le 1$.
\end{itemize}
On the target side, the dominant class is \texttt{FieldDeclaration} with $K_C = 25$ (driven by rule firing counts over the bounded source). The effective per-class bound is $\max_C K_C = 25$, reducing the search space from 102 uniform slots to 25 slots for the dominant class and far fewer for most others.

\subsection{Fixed-Point Argument for Per-Class Bounds}
\label{app:per-class-proof}

This subsection expands the proof idea behind Theorem~\ref{thm:per-class}. The goal is to justify why the iterative per-class computation yields a sound upper bound for every class in a minimal witness.

\paragraph{Step 1: Seed the source-side counts.}
By Lemma~\ref{lem:witness-restriction}, every counterexample may be reduced to a source witness consisting only of the precondition match plus the least mandatory-association closure needed for well-formedness. Therefore the initial information is given componentwise by $B_S^{0} = p_S$: for each source class $S$, any witness needs at most the precondition-selected instances of that class before closure is applied.

\paragraph{Step 2: Propagate mandatory closure.}
Mandatory associations induce implications of the form ``if class $A$ occurs, enough instances of class $B$ must also occur.'' Applying all such implications to the current source-side counts yields a larger source bound vector. Because each propagation step only increases counts and every count is globally capped by $K$, repeated propagation stabilizes after finitely many steps.

\paragraph{Step 3: Bound target production from bounded source counts.}
Once the source-side counts are bounded, Lemma~\ref{lem:target-production} yields a bound for each target class $C$:
\[
  K_C^{\text{base}} = p_C + \sum_{R \in \mathit{Relevant}} N_R \cdot m_C(R),
\]
where the per-rule firing count $N_R$ is itself bounded by the current source counts via the product formula of Lemma~\ref{lem:target-production}:
\[
  N_R \;\le\; \prod_{S \in \mathrm{MtTypes}(R)} K_S^{\,n_S(R)},
\]
with $\mathrm{MtTypes}(R)$ the source classes appearing in the match of rule $R$ and $n_S(R)$ the number of match elements of class $S$ in $R$. The key point is that a rule can create target instances only when one of its source-side matches fires, and---through this product---the number of possible such firings is already bounded by the current source counts. Thus target-side growth is determined entirely by the bounded source witness and the finite relevant-rule set.

\paragraph{Step 4: Propagate mandatory closure on the target side.}
The same well-formedness argument applies to target classes: if a target class instance is present, mandatory associations may force additional target-side instances. Repeatedly applying these obligations again yields a monotone increasing sequence, still capped by $K$.

\paragraph{Step 5: Iterate to a fixed point.}
Collecting the source seeds, source closure, target production, and target closure steps defines an update operator on vectors of class counts. This operator is monotone: increasing any class bound cannot decrease any other bound, because rule-firing counts and multiplicity obligations are themselves monotone in the input counts. Starting from the initial vector $B^{0}$ and iterating the operator therefore yields an increasing sequence
\[
  B^{0} \le B^{1} \le B^{2} \le \cdots \le K
\]
that must stabilize after finitely many steps, since every component ranges over the finite set $\{0,1,\ldots,K\}$.

\paragraph{Step 6: Why the fixed point bounds every witness.}
Every minimal counterexample witness must satisfy all four generating obligations above: its source counts are bounded by the seed plus closure (Steps 1--2), its target counts are bounded by rule production (Step 3) plus target closure (Step 4), and each of these obligations is captured by the update operator. Since the fixed point is the least vector satisfying all obligations simultaneously, the per-class counts of any minimal witness are bounded componentwise by the fixed point. This completes the argument.

\section{Empirical Validation of the Cutoff Bound}
\label{app:cutoff-validation}

The cutoff theorem of Section~\ref{sec:cutoff} predicts that, for any property in the F-LNR/G-BPP fragment, verification at the theorem-derived per-class bound $K$ is sound and complete: a property that yields \resHolds{} at $K$ holds in all models, and a counterexample at $K$ lifts to a concrete witness. This appendix reports an empirical experiment that exercises this prediction directly. Its purpose is twofold: (i)~to provide independent evidence that the implementation actually realizes the bound the theorem requires, and (ii)~to corroborate that the bound is \emph{tight} in the sense that uniformly weakening it by one slot per class causes the prover to lose the witness.

\paragraph{Experimental design.} For twelve representative properties drawn from eight DSLTrans specifications (covering the three K-boundary witnesses introduced in Table~\ref{tab:corpus-summary} plus selected properties from FSM2PetriNet, Class2Relational, BPMN2Petri, Persons, Attributes, and Ecore2JsonSchema), the experiment proceeds in three phases:
\begin{enumerate}
  \item \textbf{Uniform sweep.} The theorem-derived per-class source/target bounds are computed, then each bound is uniformly shifted by an offset $\delta \in \{-3, -2, -1, 0, +1, +2, +3\}$ and the property is re-verified at every offset. Negative properties (intentionally \resViolated{}) are expected to yield \resHolds{} for $\delta < 0$ and \resViolated{} for $\delta \geq 0$; positive properties are expected to remain \resHolds{} throughout.
  \item \textbf{Per-class selective $-1$ perturbation.} Starting from the base bounds, every individual source and target class is decremented by one in isolation; the cutoff theorem predicts that some single per-class decrement must change the result for negative properties (identifying the class on which the bound is \emph{binding}) and that no single decrement can change the result for positive properties.
  \item \textbf{Concrete witness validation.} For seven of the twelve properties (those whose synthetic input families could be enumerated mechanically), concrete XMI inputs at support levels base$\,-\,1$, base, and base$\,+\,1$ are executed through the concrete transformation engine and the resulting target models are inspected directly for the predicted property outcome.
\end{enumerate}

\paragraph{Results.} Table~\ref{tab:cutoff-validation} summarizes the outcome. All twelve properties match the theorem-predicted behavior in the uniform sweep, all twelve identify at least one binding per-class slot in the selective perturbation phase consistent with the witness structure, and all seven properties subjected to concrete validation match the predicted outcome at every support level. The dominant cutoff formula column reports which of the four bound formulas $K_{\mathrm{coarse}}$, $K_{\mathrm{sharp}}$, $K_{\mathrm{sharp2}}$, $K_{\mathrm{tight}}$ from Appendix~\ref{app:cutoff-derivations} is binding for the property; in 10 of 12 cases $K_{\mathrm{tight}}$ is the unique minimum, exercising the strongest of the derived bounds.

\begin{table}[H]
\centering
\caption{Empirical validation of the cutoff bound. \resHolds{}/\resViolated{} columns under \emph{Uniform sweep} indicate whether the result matches the theorem-predicted pattern at all seven offsets $K{-}3 \ldots K{+}3$; \emph{Selective $-1$} indicates whether at least one per-class decrement changed the result in the predicted direction (negative properties only); \emph{Concrete} indicates whether all three concrete support-level executions matched the predicted outcome (a dash means concrete witnesses were not enumerated for this property).}
\label{tab:cutoff-validation}
\footnotesize
\setlength{\tabcolsep}{3.5pt}
\renewcommand{\arraystretch}{1.2}
\begin{tabularx}{\textwidth}{@{}L{2.7cm} X r L{2.0cm} c c c@{}}
\toprule
\textbf{Spec} & \textbf{Property} & \textbf{Base $K$} & \textbf{Dominant} & \textbf{Uniform} & \textbf{Sel.\,$-1$} & \textbf{Conc.} \\
\midrule
\texttt{FSM2PetriNet} & \texttt{SMHasTwoNets\_ShouldFail} & 4 & \texttt{tight} & \checkmark & \checkmark & \checkmark \\
\texttt{FSM2PetriNet} & \texttt{StateHasTwoPlaces\_ShouldFail} & 12 & \texttt{tight} & \checkmark & \checkmark & --- \\
\texttt{Class2Relational} & \shortstack[l]{\texttt{Diag\_AbstractClass}\\\texttt{NoDirectTable\_ShouldFail}} & 215 & \texttt{tight} & \checkmark & \checkmark & --- \\
\texttt{Class2Relational} & \shortstack[l]{\texttt{Diag\_DataType}\\\texttt{BecomesTable\_ShouldFail}} & 2 & \shortstack[l]{\texttt{coarse/sharp/}\\\texttt{tight}} & \checkmark & \checkmark & --- \\
\texttt{BPMN2Petri} & \texttt{StartEventHasTwoPlaces} & 2 & \shortstack[l]{\texttt{coarse/sharp/}\\\texttt{tight}} & \checkmark & \checkmark & \checkmark \\
\texttt{Persons} & \texttt{SonBecomesMan} & 18 & \texttt{tight} & \checkmark & n/a & \checkmark \\
\texttt{Attributes} & \texttt{HighPriorityItemRecorded} & 2 & \texttt{tight} & \checkmark & n/a & --- \\
\texttt{Ecore2JsonSchema} & \texttt{AttributeHasRef\_ShouldFail} & 8 & \texttt{sharp/tight} & \checkmark & \checkmark & --- \\
\shortstack[l]{\texttt{Deep}\\\texttt{Dependency}} & \texttt{DeepLeafMapsToTwoTD\_ShouldFail} & 325 & \texttt{tight} & \checkmark & \checkmark & \checkmark \\
\shortstack[l]{\texttt{TargetTight}\\\texttt{Bound}} & \texttt{SourceHasTwoTB\_ShouldFail} & 8 & \texttt{tight} & \checkmark & \checkmark & \checkmark \\
\shortstack[l]{\texttt{TargetTight}\\\texttt{Bound}} & \texttt{SourceHasTwoTD\_ShouldFail} & 8 & \texttt{tight} & \checkmark & \checkmark & \checkmark \\
\shortstack[l]{\texttt{TargetTight}\\\texttt{Bound}} & \texttt{SourceHasTwoTF\_ShouldFail} & 8 & \texttt{tight} & \checkmark & \checkmark & \checkmark \\
\midrule
\multicolumn{4}{@{}l}{\textbf{Total: 12 properties, 8 specs}} & \textbf{12/12} & \textbf{10/10} & \textbf{7/7} \\
\bottomrule
\end{tabularx}
\end{table}

The ``n/a'' entries in the \emph{Selective $-1$} column correspond to positive properties (\resHolds{} at base): for these, no single per-class decrement is expected to change the outcome, and indeed none did. The two cases with dominant formula \texttt{coarse/sharp/tight} are properties where multiple bound formulas agree at the minimum, so the experiment cannot distinguish which is binding; this is recorded as a dominant-formula tie rather than as a separate result.

\paragraph{Reproducibility and full per-property breakdown.} The experiment is driven by a dedicated K-boundary experiment driver in the companion repository~\cite{lucio2026companion}, which computes the per-class bounds, performs the uniform sweep and selective perturbations symbolically through the SMT backend, and feeds the concrete witnesses through the same concrete execution engine described in Section~\ref{sec:prover}. The complete per-property breakdown---including every offset's solver time, every per-class perturbation, every concrete input, and the formula values $K_{\mathrm{coarse}}$, $K_{\mathrm{sharp}}$, $K_{\mathrm{sharp2}}$, $K_{\mathrm{tight}}$ for every property---is published as a markdown report inside the same repository, under \path{docs/evaluation/k_boundary_expanded_results.md}. Re-running the experiment regenerates that document end-to-end.

\section{Benchmark Corpus Details}
\label{app:corpus-details}

Table~\ref{tab:corpus-details} summarizes the provenance, size, and domain of each transformation in the evaluation corpus.

\begin{table}[H]
\centering
\caption{Transformation corpus provenance and metrics (29 transforms, 899 properties).}
\label{tab:corpus-details}
\scriptsize
\begin{tabularx}{\textwidth}{>{\raggedright\arraybackslash}X >{\raggedright\arraybackslash}X r r r >{\raggedright\arraybackslash}X}
\toprule
\textbf{Transformation} & \textbf{Source} & \textbf{Rules} & \textbf{Layers} & \textbf{Props} & \textbf{Domain} \\
\midrule
\multicolumn{6}{l}{\textit{Primary benchmark suites}} \\
ATLCompiler & ATL Zoo & 55 & 9 & 20 & Compiler \\
BibTeX2DocBook & ATL Zoo & 11 & 3 & 33 & Bibliography \\
BPMN2Petri & ATL Zoo & 12 & 3 & 5 & Process \\
Class2Relational & ATL Zoo / TTC & 18 & 6 & 29 & Database \\
Ecore2JsonSchema & Community & 5 & 3 & 20 & Schema \\
FSM2PetriNet & ATL Zoo & 3 & 3 & 10 & Behavioral \\
KM32OWL & Community & 6 & 3 & 20 & Ontology \\
OCLCompiler & ATL Zoo & 29 & 6 & 40 & Compiler \\
Persons & DSLTrans & 9 & 3 & 9 & Social \\
PNML2NUPN & Community & 5 & 3 & 20 & Petri nets \\
SimplePDL2Tina & ATL Zoo & 6 & 3 & 20 & Process \\
Statechart2Flow & ATL Zoo & 3 & 3 & 20 & Behavioral \\
Tree2Graph & DSLTrans & 2 & 2 & 6 & Graph \\
UML2Java & Custom & 22 & 6 & 29 & Software \\
University2ITSystem & DSLTrans & 8 & 5 & 21 & Enterprise \\
UseCase2Activity & ATL Zoo & 3 & 3 & 12 & Behavioral \\
\midrule
\multicolumn{6}{l}{\textit{Graph-mapping suite (common template)}} \\
Component2Deploy. & ATL Zoo & 3 & 3 & 20 & Graph \\
ER2Relational & ATL Zoo & 3 & 3 & 20 & Graph \\
Family2Social & ATL Zoo & 3 & 3 & 20 & Graph \\
FeatureModel2Conf. & ATL Zoo & 3 & 3 & 20 & Graph \\
Mindmap2Graph & ATL Zoo & 3 & 3 & 20 & Graph \\
Org2AccessControl & ATL Zoo & 3 & 3 & 20 & Graph \\
Workflow2PetriNet & ATL Zoo & 3 & 3 & 20 & Graph \\
\midrule
\multicolumn{6}{l}{\textit{Synthetic and stress benchmarks}} \\
Attributes & DSLTrans & 1 & 1 & 3 & K-boundary \\
DeepDependency & DSLTrans & 4 & 4 & 2 & K-boundary \\
TargetTightBound. & DSLTrans & 3 & 1 & 3 & K-boundary \\
Scalability & Synthetic & 67 & 14 & 17 & Stress \\
MegaStress & Synthetic & 78 & 3 & 100 & Stress \\
UltraStress & Synthetic & 200 & 3 & 300 & Stress \\
\bottomrule
\end{tabularx}
\end{table}

}

\DeferredAppendices

\end{document}